%% file: main.tex
\documentclass[11pt,letterpaper]{article}
\input{config/preamble}

\input{config/metadata}
\renewcommand{\PaperReadingBreak}{\clearpage}
\renewcommand{\PaperFloatBarrier}{\FloatBarrier}

\begin{document}
\hypersetup{pageanchor=false}
\begin{titlepage}
\renewcommand{\thefootnote}{\fnsymbol{footnote}}
\centering
\vspace*{0.35in}
{\LARGE\bfseries\PaperTitle\par}
\vspace{0.45in}
\PaperAuthorBlock\par
\vspace{0.35in}
\begin{minipage}{0.92\textwidth}
\begin{center}\textbf{Abstract}\end{center}
\input{sections/00-abstract}
\end{minipage}\par
\vfill
{\DraftDate\par}
\end{titlepage}

\hypersetup{pageanchor=true}
\pagenumbering{roman}
\setcounter{tocdepth}{1}
\tableofcontents
\clearpage
\pagenumbering{arabic}
\input{body}

\label{paper:main-end}

\clearpage
\bibliographystyle{alphaurl}
\bibliography{references}
\end{document}

%% file: config/preamble.tex
\usepackage[T1]{fontenc}
\usepackage[utf8]{inputenc}
\usepackage[margin=1in,letterpaper]{geometry}
\usepackage{amsmath,amsthm,mathtools,aliascnt}
\usepackage{newtxtext,newtxmath}
\usepackage{microtype}
\usepackage{booktabs,tabularx,array}
\usepackage{enumitem}
\usepackage{graphicx,tikz}
\usetikzlibrary{arrows.meta,positioning}
\usepackage{xcolor,xurl}
\usepackage{needspace,placeins}
\usepackage[hidelinks,pdfencoding=auto,psdextra]{hyperref}
\usepackage[nameinlink,capitalise,noabbrev]{cleveref}

\setlist{topsep=0.4em,itemsep=0.2em,parsep=0pt,leftmargin=1.8em}
\input{config/macros}

%% file: config/macros.tex
\newcolumntype{Y}{>{\raggedright\arraybackslash}X}

\providecommand{\PaperTheoremStyle}{plain}
\providecommand{\PaperDefinitionStyle}{definition}
\providecommand{\PaperRemarkStyle}{remark}
\theoremstyle{\PaperTheoremStyle}
\newtheorem{theorem}{Theorem}[section]
\newaliascnt{lemma}{theorem}
\newtheorem{lemma}[lemma]{Lemma}
\aliascntresetthe{lemma}
\newaliascnt{corollary}{theorem}
\newtheorem{corollary}[corollary]{Corollary}
\aliascntresetthe{corollary}
\theoremstyle{\PaperDefinitionStyle}
\newaliascnt{definition}{theorem}

\aliascntresetthe{definition}
\newaliascnt{example}{theorem}

\aliascntresetthe{example}
\theoremstyle{\PaperRemarkStyle}
\newaliascnt{remark}{theorem}

\aliascntresetthe{remark}
\crefname{appendix}{Appendix}{Appendices}
\Crefname{appendix}{Appendix}{Appendices}

\newcommand{\N}{\mathbb{N}}
\newcommand{\Z}{\mathbb{Z}}
\newcommand{\BT}{\operatorname{BT}}
\newcommand{\ord}{\operatorname{ord}}
\newcommand{\val}{\operatorname{val}}
\newcommand{\EqH}{\mathsf{Eq}_{\mathrm{HORS}}}
\newcommand{\EqFour}{\mathsf{Eq}_{4,\Sigma_2}}
\newcommand{\EqProd}{\mathsf{Eq}^{\mathrm{prod}}_{4,\Sigma_2}}
\newcommand{\EqSafeProd}{\mathsf{Eq}^{\mathrm{safe,prod}}_{4,\Sigma_2}}
\newcommand{\Safe}{\operatorname{Safe}}
\newcommand{\PiOne}{\Pi^0_1}

\newcommand{\NN}{\mathsf{N}}
\newcommand{\PP}{\mathsf{P}}
\newcommand{\Pref}{\operatorname{Pref}}
\newcommand{\Prod}{\operatorname{Prod}}
\newcommand{\Valid}{\operatorname{Valid}}
\newcommand{\Num}{\operatorname{Num}}
\newcommand{\Pol}{\operatorname{Pol}}
\newcommand{\Reg}{\operatorname{Reg}}
\newcommand{\Obs}{\operatorname{Obs}}
\newcommand{\halts}{\mathord{\downarrow}}
\newcommand{\nonhalts}{\mathord{\uparrow}}
\newcommand{\Cand}{\mathcal C}
\newcommand{\Legal}{\mathcal L}
\newcommand{\pc}{\operatorname{pc}}
\newcommand{\dom}{\operatorname{dom}}
\newcommand{\ev}[1]{\llbracket #1\rrbracket}
\newcommand{\tyC}{\mathsf C}

\newcommand{\proofentry}[1]{\begingroup\urlstyle{tt}\nolinkurl{#1}\endgroup}

\providecommand{\PaperReadingBreak}{}
\providecommand{\PaperFloatBarrier}{}

\providecommand{\Description}[1]{}

\newcommand{\PaperTitle}{Tree Equivalence of Higher-Order Recursion Schemes Is Undecidable}

%% file: config/metadata.tex
\newcommand{\PaperAuthorNote}{Authors are listed in order of contribution rather than alphabetically, with the first author making the primary contribution.}
\newcommand{\PaperAuthorBlock}{%
  \begin{tabular}{@{}c@{\hspace{2.5em}}c@{}}
    {\large Mingkun Xiao\footnotemark} & {\large Yixuan Sun}\\[0.5em]
    Capital Normal University & Kean University\\[0.35em]
    \href{mailto:1230604016@cnu.edu.cn}{\nolinkurl{1230604016@cnu.edu.cn}} &
    \href{mailto:sunyixua@kean.edu}{\nolinkurl{sunyixua@kean.edu}}
  \end{tabular}%
  \footnotetext{\PaperAuthorNote}}
\newcommand{\DraftDate}{October 27, 2026}
\hypersetup{
  pdftitle={\PaperTitle},
  pdfauthor={Mingkun Xiao and Yixuan Sun},
  pdfsubject={Compiler and finite-semantics proofs for HORS tree equivalence},
  pdfkeywords={higher-order recursion schemes, tree equivalence, undecidability, computability}
}

%% file: sections/00-abstract.tex
Higher-order recursion schemes are finite, simply typed recursive programs
that generate potentially infinite trees. Their monadic second-order
model-checking problem is decidable. We prove that equality of their
generated trees is undecidable and, more precisely, $\PiOne$-complete
under total computable many-one reductions. The lower bound holds for
safe, productive schemes of order at most four over the fixed ranked
alphabet $\Sigma_2=\{e:0,d:2\}$, consisting of one nullary and one
binary symbol.

Adapting the counter-guard witness of Boiret, Pi\'orkowski, and Schmude, we
generate all candidate executions of a two-counter machine by producing
both continuations at every counter test. An integer weight is nonzero
exactly on legal prefixes. Functions at fixed simple types represent
the polynomial of growing degree. A shift recurrence updates them
without storing a separate step counter, allowing the witness to be
implemented by safe applicative HORS rules.
At candidate halts, the intermediate schemes generate unary observation subtrees
whose lengths encode the two natural-number components of the signed weight.
Invalid candidates give matching observations; a legal halt produces
a finite discrepancy. Fixed constructor contexts encode these
intermediate trees over $\Sigma_2$.

For the upper bound, a finite monotone semantics decides head convergence
and supports an algorithm computing exact tree prefixes at any
prescribed finite depth, including positions labelled $\bot$.

%% file: body.tex
\input{sections/01-introduction}
\input{sections/technical-overview}
\label{paper:overview-end}
\PaperReadingBreak
\input{sections/02-preliminaries}
\input{sections/03-witness}
\input{sections/04-compiler}
\input{sections/05-upper-bound}
\input{sections/07-consequences}
\input{sections/06-formalization}

%% file: sections/01-introduction.tex
\section{Introduction}\label{sec:introduction}

Higher-order recursion schemes (HORS) are finite, simply typed recursive
programs generating potentially infinite, ordered, labelled trees.
Their parameters may be functions, while terminal symbols act as tree
constructors.
HORS also generate binder-free encodings of the $\eta$-long B\"ohm trees
of $\lambda Y$-terms, with binding information represented by
de Bruijn levels~\cite[Theorem~11, p.~96]{clairambault2013}.

Given such a finite description, a basic verification task is to decide
whether the generated tree satisfies a specification. Monadic
second-order logic (MSO) provides a language for these specifications,
with quantification over nodes and sets of nodes. Given a scheme and
an MSO sentence, one can decide whether its generated tree satisfies
the sentence~\cite[Theorem~1]{ong2006}. This model-checking result supports
verification of higher-order functional programs~\cite{kobayashi2009}.

A second task is to compare two descriptions. Tree equivalence asks
whether two independently given schemes generate the same tree.
At order one it is decidable via language equivalence for deterministic
pushdown automata~\cite{senizergues1997,jochems2021}.
The higher-order case has long remained
open~\cite{clairambault2013,ong2015,jochems2021}.

For nondeterministic pushdown processes, Broadbent and G\"oller prove
that strong bisimilarity is undecidable already for $\varepsilon$-free
order-two automata~\cite[Theorem~5, p.~167]{broadbent2012}.
Their construction uses the choice among transitions with the same
action label. In deterministic ranked-tree equality, child addresses
fix which subtrees must be compared. The nondeterministic bisimilarity
lower bound therefore does not settle this equality problem.

We prove that exact tree equality is undecidable already when both
schemes are safe and productive, their type order is at most four, and their
terminal alphabet consists of one nullary and one binary symbol.
In a productive scheme, every output node can be exposed after finitely
many unfolding steps. Safety imposes a standard syntactic discipline on
types and parameter dependencies~\cite[Definition~1, p.~8]{hague2017}.
In particular, a function passed as an argument may depend only on rule
parameters whose orders are at least its own; \cref{sec:preliminaries}
gives the full definition.
The construction uses a fixed finite collection of simple types,
independent of the source machine and the length of a candidate execution.

\subsection{Main result and a scheme example}

For a scheme $G$, write $\BT(G)$ for the tree it generates, with $\bot$
marking a position where no terminal head is ever exposed. A productive scheme has no
$\bot$ nodes; finite leaves are allowed. Equality compares labels at
the same child addresses, retaining child order and ignoring the
number of unfolding steps used to expose a node.

\begin{theorem}[Tree equivalence]\label{thm:main}
Tree equivalence of deterministic higher-order recursion schemes is
undecidable and, more precisely, $\PiOne$-complete (co-c.e.-complete) under total
computable many-one reductions. This remains true when both schemes
are required simultaneously to be safe and productive, to have order at most
four, and to use the fixed ranked alphabet
\[
                    \Sigma_2=\{e:0,\ d:2\}.
\]
\end{theorem}

Two consequences concern other models of computation.
The effective, order-preserving correspondence for safe
schemes~\cite[Theorems~5.1 and~5.3]{knapik2002}
(see also~\cite[Theorem~2, p.~14]{hague2017}) gives undecidability of
ranked-tree equality for deterministic order-four higher-order
pushdown automata \emph{without collapse} (\cref{cor:pda-tree}).
The translations of~\cite[Theorem~11 and Corollary~19]{clairambault2013}
give $\PiOne$-completeness of $\eta$-long B\"ohm-tree equality for
simply typed $\lambda Y$-terms, up to $\alpha$-equivalence
(\cref{cor:lambda-y}). The latter transfer makes no claim that the
same order bound or safety restriction is preserved.

For the construction, we use the intermediate alphabet
$\Sigma_0=\{z:0,a:1,b:3\}$: $z$ ends an observation, $a$ counts one
unit, and $b$ groups an observation with two continuations.
A routine injective encoding replaces these constructors by fixed
$\Sigma_2$-tree contexts, preserving safety, productivity, and the
order-four bound. This is the final pass of the compiler;
\cref{sec:binary-target} proves its specification.

\Cref{fig:scheme-example} gives a scheme with types
$\mathsf{S}:o$, $\mathsf{F}:(o\to o)\to o$, and $\mathsf{D}:(o\to o)\to o\to o$.
The parameter $f$ denotes a unary context, and $\mathsf{D}\,f$ composes that
context with itself. Passing it to the next call of $\mathsf{F}$ gives observation
subtrees $az,a^2z,a^4z,\ldots$ along a right spine, where $a^nz$
is a chain of $n$ nodes labelled $a$ ending in the leaf $z$.
Their unbounded distinct lengths make the resulting tree nonregular,
although its description consists of three rules of order at most two.

\input{figures/scheme-example}

The main theorem also rules out a computable input-size bound on the
depth of a distinguishing node for unequal pairs in this class
(\cref{cor:depth}).

\PaperFloatBarrier

Putting two output trees below a common root does not by itself
compare labels at corresponding addresses. More generally, no total
computable construction can map each pair of safe, productive schemes
of order at most four over $\Sigma_2$ to a HORS whose tree satisfies
a fixed MSO sentence exactly when the input trees differ. Composing
such a construction with an MSO model checker would decide the
equivalence problem in \cref{thm:main}.

\subsection{Contributions}

Boiret, Pi\'orkowski, and Schmude~\cite[Section~4, Theorem~20, p.~48:14]{boiret2018}
prove undecidability for deterministic polynomial-register automata with
substitution, even on word inputs. We adapt their counter-guard witness,
retaining the polynomial $Z_m(X)=\prod_{j=1}^m(X-j)$, the multipliers
$Z_m(c)$ and $c$ for zero and positive guesses, and the first-error
argument. \Cref{sec:witness} specifies the different control checks
and resulting weights.

The compiler realizes arithmetic on pairs of natural numbers, evaluation
of the witness polynomial, and its shift update. Its safe, productive
HORS implementation has three parts:
\begin{enumerate}
\item State calls enforce control flow and identify halts, removing
the register construction's polynomial tests of control states.
The weight checks only the selected counter guard.
\item Functions store the evolving polynomial at fixed simple types.
A shift recurrence constructs new components from the old ones without
retaining the step count as a lower-order parameter. The stored polynomial
terms satisfy the operand condition for safety as the degree grows.
\item Tree equality expresses cancellation. A signed weight
$W=W^+-W^-$ is zero exactly when its natural components are equal,
even if both are large. Generating a unary observation subtree for each
component expresses cancellation as equality of two varying trees.
\end{enumerate}

\subsection{Proof architecture}

Let $M$ be a deterministic two-counter Minsky machine.
Write $M\halts$ for halting from the first instruction with both
counters zero, and $M\nonhalts$ for $\neg(M\halts)$.
This zero-start halting problem is undecidable~\cite{dudenhefner2022}.
\Cref{sec:overview-candidates} gives the instruction semantics;
\cref{sec:preliminaries} fixes the source decision problem.
The first compiler stage produces $G_M^L,G_M^R$ over the intermediate
alphabet $\Sigma_0=\{z:0,a:1,b:3\}$.
Its final pass replaces the constructors by fixed tree contexts $E$
and outputs $\widehat G_M^\diamond=(G_M^\diamond)^E$ over $\Sigma_2$,
for $\diamond\in\{L,R\}$. We prove
\begin{equation}\label{eq:main-reduction}
 M\nonhalts\quad\Longleftrightarrow\quad
                 \BT(\widehat G_M^L)=\BT(\widehat G_M^R).
\end{equation}
In the intermediate trees, each zero test generates two continuation subtrees:
one follows the instruction for a zero counter, the other the instruction
for a positive counter. Both are generated without checking which
answer is correct. A path through these continuations describes a
\emph{candidate execution}. Its guessed answers may be wrong; a finite
prefix is \emph{legal} when every guess agrees with the counter value
at that step. Wrong guesses can even make the candidate's counters negative.

Along each candidate, arithmetic updates maintain an integer weight
$W$ that is nonzero exactly on legal prefixes. We represent it as
$W=W^+-W^-$ with natural components. If the candidate reaches a halting
position, the left and right schemes generate the unary observation
subtrees $a^{W^+}z$ and $a^{W^-}z$, respectively.
These observations agree exactly when $W=0$.
An illegal candidate therefore gives matching observations, while a
legal halt gives a finite discrepancy. Outside these observations,
the two trees have matching labels and branching, so every discrepancy
comes from a legal halt.

We prove that the stored functions can be reused in arithmetic and
that every candidate branch remains productive. An address argument shows that
all mismatches occur at halt observations. The technical
overview follows a two-instruction machine through the witness updates
and HORS rules to a distinguishing address.

Jochems's effective transformation replaces every $\bot$ node by an
infinite chain of a fresh terminal, preserving and reflecting tree
equality~\cite[Lemma~2.7 and Proposition~2.8, p.~42]{jochems2021}.
Its outputs are productive, so enumerating their finite prefixes
semidecides inequality and yields the general co-c.e.\ upper bound.
We give a direct finite-semantics adequacy proof~\cite[Theorem~4.1, p.~10]{salvati2015}
and an algorithm computing exact prefixes at any prescribed depth.

\paragraph{Organization.}
After the technical overview, \cref{sec:preliminaries} fixes the
standard syntax, observations, and source convention.
\Cref{sec:witness} defines candidate executions and their arithmetic witness.
\Cref{sec:compiler} proves the compiler's correctness, typing, safety,
and two-terminal output.
\Cref{sec:upper} supplies the matching upper bound, and
\cref{sec:consequences} derives the mathematical consequences.

%% file: figures/scheme-example.tex
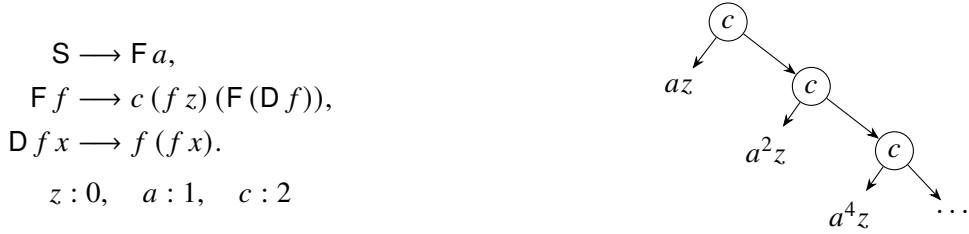
\begin{figure*}[htbp]
  \centering
  \begin{minipage}[c]{0.43\textwidth}
    \[
    \begin{aligned}
      \mathsf{S} &\longrightarrow \mathsf{F}\,a,\\
      \mathsf{F}\,f &\longrightarrow c\,(f\,z)\,(\mathsf{F}\,(\mathsf{D}\,f)),\\
      \mathsf{D}\,f\,x &\longrightarrow f\,(f\,x).
    \end{aligned}
    \]
    \centering
    $z:0,\quad a:1,\quad c:2$
  \end{minipage}
  \hfill
  \begin{minipage}[c]{0.53\textwidth}
    \centering
    \begin{tikzpicture}[>=Stealth, every node/.style={font=\normalsize},
      level/.style={circle,draw,inner sep=2.5pt}]
      \node[level] (c0) at (0,0) {$c$};
      \node (l0) at (-0.65,-0.85) {$az$};
      \node[level] (c1) at (1.1,-0.85) {$c$};
      \node (l1) at (0.5,-1.7) {$a^2z$};
      \node[level] (c2) at (2.2,-1.7) {$c$};
      \node (l2) at (1.6,-2.55) {$a^4z$};
      \node (more) at (3.0,-2.55) {$\cdots$};
      \draw[->] (c0) -- (l0);
      \draw[->] (c0) -- (c1);
      \draw[->] (c1) -- (l1);
      \draw[->] (c1) -- (c2);
      \draw[->] (c2) -- (l2);
      \draw[->] (c2) -- (more);
    \end{tikzpicture}
  \end{minipage}
  \caption{A finite order-two scheme and its unique output tree.
  Each expression $a^nz$ abbreviates a finite unary subtree. Passing $\mathsf{D}\,f$ to the
  recursive call doubles the context used for the next observation.}
  \label{fig:scheme-example}
  \Description{Three recursive rules generate a right spine of binary c nodes.
  Their left children are unary a chains ending in z, of lengths one, two,
  four, and successive powers of two.}
\end{figure*}

%% file: sections/technical-overview.tex
\section{Technical overview}\label{sec:overview}

Given a two-counter machine $M$, we construct two safe, productive HORS
of order at most four whose trees agree exactly when $M$ does not halt.
The first stage uses $\Sigma_0=\{z:0,a:1,b:3\}$ to separate observations
from continuations; fixed constructor contexts then give the final
alphabet $\Sigma_2=\{e:0,d:2\}$.
The counter-guard witness and its first-error argument are adapted from
Boiret, Pi\'orkowski, and Schmude~\cite[Section~4]{boiret2018}.
The construction below stores the evolving polynomial at fixed types
and expresses a zero weight by comparing two unary observation trees.

\subsection{Candidate executions and the adapted witness}\label{sec:overview-candidates}

We use deterministic two-counter Minsky machines with instruction forms
\[
 \iota::=\mathsf{inc}_A\mid\mathsf{inc}_B
       \mid\mathsf{dec}_A(j)\mid\mathsf{dec}_B(j),
 \qquad j\in\N.
\]
A program is a finite list $M=[\iota_1,\ldots,\iota_L]$, with
$L=|M|$ and instruction positions $I_M=\{1,\ldots,L\}$.
The lookup $M(i)=\iota_i$ is defined exactly for $i\in I_M$.
Jump targets may be any natural number, including zero and positions
beyond $L$.

For $X\in\{A,B\}$, an instruction $\mathsf{inc}_X$ increments counter
$X$ and advances to the next instruction. An instruction
$\mathsf{dec}_X(j)$ advances when $X=0$; otherwise it decrements $X$
and jumps to $j$. The other counter is unchanged.
Execution starts at position $1$ with both counters zero and halts
exactly when the program counter is outside $I_M$; no separate halt
instruction is needed. The empty program has $I_M=\varnothing$ and
halts immediately.

To construct the candidate tree, unfold this finite control from the
initial configuration. At each test, generate both successors without
checking their guards, labelled $1$ for the zero choice and $2$ for
the positive choice. An increment gives two copies of its successor;
a halt gives a leaf. A path describes a \emph{candidate execution},
whose virtual counters may become negative after a wrong positive
choice. \Cref{sec:candidate-semantics} defines this unfolding and its
relation to the source steps.

Write $\Cand_M\subseteq\{1,2\}^*$ for these finite candidate prefixes,
including $\epsilon$. Let $\Legal_M\subseteq\Cand_M$ be the set
of prefixes on which every guessed guard holds.
For $v\in\Cand_M$, let $\pc_M(v)$ be its final program counter.

The unfolding fixes the control flow, so the weight only needs to detect
incorrect counter guards. After $m$ instructions of a legal execution,
each counter lies in $[0,m]$. The polynomial
\[
 Z_m(X)=\prod_{j=1}^{m}(X-j),\qquad Z_0(X)=1,
\]
satisfies $Z_m(0)=(-1)^m m!\ne0$ and $Z_m(c)=0$ for $1\le c\le m$.
Initialize $W=1$. For the current instruction $\iota$ on counter
$X\in\{A,B\}$, let $c$ be its old value and $e\in\{1,2\}$ the
candidate choice: at a test, $e=1$ guesses zero and $e=2$ guesses positive.
Update the selected counter and weight simultaneously:
\begin{equation}\label{eq:overview-update}
 (c',W')=
 \begin{cases}
  (c+1,W),&\iota=\mathsf{inc}_X,\\
  (c,WZ_m(c)),&\iota=\mathsf{dec}_X(j)\land e=1,\\
  (c-1,Wc),&\iota=\mathsf{dec}_X(j)\land e=2.
 \end{cases}
\end{equation}
All right-hand sides use the old values. Every instruction also advances
$m$ and replaces $Z_m$ by $Z_{m+1}$.
Write $W_M(v)$ for the resulting weight of a prefix $v$; in particular,
$W_M(\epsilon)=1$.

Correct choices multiply the weight by a nonzero integer. At the first
wrong choice the tested counter is still in $[0,m]$: a false zero
choice has $1\le c\le m$ and multiplies by $Z_m(c)=0$; a false positive
choice has $c=0$ and also multiplies by zero. Later updates preserve zero even if the virtual
counters become negative. Consequently
\begin{equation}\label{eq:overview-witness}
 \forall v\in\Cand_M,\qquad
 W_M(v)\ne0\quad\Longleftrightarrow\quad v\in\Legal_M.
\end{equation}
This is the first-error argument of the polynomial-register construction.
\Cref{sec:witness} proves the invariant and its correspondence with
the precise source-machine convention.

To realize the witness in HORS, we represent signed integers by pairs
of natural numbers and the evolving polynomial by functions at fixed
types. Comparing unary subtrees generated from the two integer
components expresses the zero condition. The polynomial-register
construction uses polynomial-substitution primitives; our function
representation implements the required evaluations and updates by
ordinary HORS application.

\subsection{Encoding cancellation by comparison}

Represent an integer by a pair $(u,v)\in\N^2$ with value $u-v$.
The pair is not normalized. Increment, decrement, and multiplication
are realized by
\[
\begin{aligned}
 \operatorname{inc}(u,v)&=(u+1,v),&
 \operatorname{dec}(u,v)&=(u,v+1),\\
 (u,v)\otimes(r,s)&=(ur+vs,us+vr).&&
\end{aligned}
\]
The multiplication components have difference
$(u-v)(r-s)$. All component operations use successor, addition, and
multiplication on naturals. A negative virtual counter after an
invalid choice still has two natural components, and decrement requires
no predecessor operation.

Keep the weight as $W=W^+-W^-$. In each intermediate output tree, a candidate node
has label $b$: child $0$ carries an observation, and children $1,2$
carry the continuations. At nonhalting candidates the observation is $z$.
At a candidate halt, the left scheme generates the observation subtree
$a^{W^+}z$ and the right scheme generates $a^{W^-}z$.
Each is a finite chain of $a$-nodes ending in $z$; the continuation
children then use the shared padding defined in \cref{sec:overview-control}.
The two observations satisfy
\begin{equation}\label{eq:overview-cancellation}
 a^{W^+}z=a^{W^-}z\quad\Longleftrightarrow\quad W=0.
\end{equation}
Equality of the generated trees supplies this comparison. For example,
the pair $(1,1)$ represents zero and gives $az$ on both sides. Comparing
its left observation with $z$ would report a discrepancy.
Thus the register construction's comparison with a constant-zero
function becomes a comparison of two varying output trees.

\subsection{A safe polynomial representation at fixed types}\label{sec:overview-order}

The polynomial $Z_m$ has degree $m$. Listing all its coefficients would
require a growing collection of coordinates. We retain two functions
$P_m^+,P_m^-:\N^2\to\N$ satisfying
\[
 P_m^+(u,v)-P_m^-(u,v)=Z_m(u-v).
\]
To represent their inputs and outputs in HORS, use the simple types
\[
 \tyC=o\to o,\qquad
 \NN=\tyC\to\tyC,\qquad
 \PP=\NN\to\NN\to\NN.
\]
A Church numeral of type $\NN$ represents iteration: the numeral for
$n$, applied to $f,x$, gives $f^n(x)$. A stored polynomial component
has type $\PP$ and maps two such numerals to a numeral for its value.
The order function
\begin{equation}\label{eq:type-order}
 \ord(o)=0,\qquad
 \ord(\sigma\to\tau)=\max\{\ord(\sigma)+1,\ord(\tau)\}
\end{equation}
gives $\ord(\tyC)=1$, $\ord(\NN)=2$, and $\ord(\PP)=3$.

Safety requires homogeneous types and a bound on the parameters of
operands~\cite[Definition~1]{hague2017}. Homogeneity means that curried
parameter orders are nonincreasing, recursively at argument types:
$\NN=\tyC\to o\to o$ has parameter orders $1,0$, and
$\PP=\NN\to\NN\to\tyC\to o\to o$ has orders $2,2,1,0$.
A term $s$ supplied as an argument must satisfy
$\ord(s)\le\ord(x)$ for every rule parameter $x$ occurring in $s$.
\Cref{sec:preliminaries} gives the full definition.

This condition constrains how we update the stored functions $p,q:\PP$.
The polynomial-register construction stores the step counter $m$ and
uses $Z_{m+1}(x)=Z_m(x)(x-m-1)$~\cite[Section~4]{boiret2018}.
A direct encoding that passes the operand
$\mathsf{Upd}\,p\,q\,t:\PP$, with a rule parameter $t:\NN$
representing $m$, violates safety: the order-three operand contains
an order-two parameter. We therefore use the shift identity
\[
 Z_{m+1}(x)=(x-1)Z_m(x-1).
\]
Writing $C_m=(P_m^+,P_m^-)$, it yields
\begin{equation}\label{eq:overview-polynomial}
 C_0(u,v)=(1,0),\qquad
 C_{m+1}(u,v)=(u,v+1)\otimes C_m(u,v+1).
\end{equation}
Taking differences proves the representation for all natural $u,v$,
including $u<v$. In particular,
\[
 P_{m+1}^+(u,v)=uP_m^+(u,v+1)+(v+1)P_m^-(u,v+1),
\]
and the negative component interchanges $P_m^+$ and $P_m^-$.
The index $m$ belongs to this mathematical invariant; the compiled
state does not carry a step-counter parameter.

Ordinary Church arithmetic implements the component operations.
The helpers $\mathsf{Pos},\mathsf{Neg}$ compute $ur+vs$ and $us+vr$;
for example,
\[
 \mathsf{Pos}\ u\ v\ r\ s\ f\ x
   \to\mathsf{Add}\,(\mathsf{Mul}\ u\ r)\,
                         (\mathsf{Mul}\ v\ s)\,f\,x.
\]
The full rules for numerals, successor, addition and multiplication
are given in \cref{sec:rules}; these helpers have order at most three.
Introduce $\mathsf{Next}_\pm:\PP\to\PP\to\PP$, of order four.
If $p,q$ represent $P_m^+,P_m^-$, then
$\mathsf{Next}_+\,p\,q$ and $\mathsf{Next}_-\,p\,q$ represent the next
components. The positive rule is
\begin{equation}\label{eq:overview-next}
\begin{aligned}
 \mathsf{Next}_+\ p\ q\ u\ v\ f\ x
  \to{}&\mathsf{Pos}\ u\ (\mathsf{Succ}\ v)\\[-2pt]
       &\quad(p\,u\,(\mathsf{Succ}\ v))\,
             (q\,u\,(\mathsf{Succ}\ v))\ f\ x.
\end{aligned}
\end{equation}
The negative update replaces $\mathsf{Pos}$ by $\mathsf{Neg}$.
Here $p,q:\PP$, $u,v:\NN$, $f:\tyC$, and $x:o$; the right-hand
side has ground type $o$. The two polynomial calls evaluate the old
components at the same shifted argument pair.

The stored operand $\mathsf{Next}_\pm\,p\,q:\PP$ depends only on the
order-three parameters $p,q$. Numeral arguments are supplied later:
$p\,u\,(\mathsf{Succ}\ v)$ has order two, at most the order of each
parameter it contains. Its order-three prefix $p\,u$ is in function
position and has no operand inequality of its own.
\Cref{sec:safety-audit} checks all operand families and homogeneous types.

Repeated updates build nested applications of these fixed nonterminals.
The terms grow in size and represent polynomials of increasing degree,
while \emph{their type remains $\PP$}; \cref{lem:polynomial} proves the invariant.

\subsection{Control rules, reusable representations, and productivity}\label{sec:overview-control}

The arithmetic terms now form the parameters of state calls that
generate the candidate tree. Each source instruction at position $i$
is compiled into the rule defining a control nonterminal $\mathsf{F}_i$;
$\mathsf{F}_0$ handles halts. A state call has the form
\[
 \mathsf{F}_i\,p\,q\,u\,v\,r\,s\,h\,k,
 \qquad \mathsf{F}_i:\PP\to\PP\to\NN^6\to o.
\]
The terms $p,q$ represent the current polynomial components.
The numeral pairs $(u,v),(r,s),(h,k)$ represent the signed values of
the counters $A,B$ and weight $W$, respectively. The start rule supplies
representations of $C_0$, zero counter pairs, and weight pair $(1,0)$.
The control type also has order four. The source program determines
the number of control rules, while their types are fixed;
\cref{lem:typing} gives the complete type audit.
Normalize source positions by
\[
 \widehat j=\begin{cases}j,&j\in I_M,\\0,&j\notin I_M.\end{cases}
\]
Thus every out-of-range program position uses the halting nonterminal $\mathsf{F}_0$.
Write
$p'=\mathsf{Next}_+\,p\,q$ and $q'=\mathsf{Next}_-\,p\,q$ as
textual abbreviations. If instruction $i$ is $\mathsf{dec}_A(j)$,
its control rule is
\begin{equation}\label{eq:overview-control}
\begin{aligned}
 \mathsf{F}_i\,p\,q\,u\,v\,r\,s\,h\,k\ \to\ b\,z\
 &\bigl(\mathsf{F}_{\widehat{i+1}}\,p'\,q'\,u\,v\,r\,s\\[-2pt]
 &\qquad(\mathsf{Pos}\ h\ k\ (p\,u\,v)\ (q\,u\,v))\\[-2pt]
 &\qquad(\mathsf{Neg}\ h\ k\ (p\,u\,v)\ (q\,u\,v))\bigr)\\
 &\bigl(\mathsf{F}_{\widehat j}\,p'\,q'\,u\,(\mathsf{Succ}\ v)\,r\,s\\[-2pt]
 &\qquad(\mathsf{Pos}\ h\ k\ u\ v)\,
          (\mathsf{Neg}\ h\ k\ u\ v)\bigr).
\end{aligned}
\end{equation}
The ternary constructor $b$ reserves child $0$ for an observation and
children $1,2$ for the zero and positive continuations. Both use the
old counter in the weight update; the zero choice also uses the old
polynomial. Both pass the new polynomial. There is one production,
so emitting both continuations introduces no nondeterminism.
At an increment the two continuation children are identical.

At a halt, the numeral parameters $h,k:\NN$ represent the weight
components $W^+,W^-$. Applying them to the constructors $a,z$ gives
ground terms $h\,a\,z$ and $k\,a\,z$, which generate the observation
chains $a^{W^+}z$ and $a^{W^-}z$.
Introduce a shared padding nonterminal $\mathsf{K}:o$ with rule
$\mathsf{K}\to b\,z\,\mathsf{K}\,\mathsf{K}$ for the two continuation children after a halt.
The left halting rule has body $b\,(h\,a\,z)\,\mathsf{K}\,\mathsf{K}$ and the right
has body $b\,(k\,a\,z)\,\mathsf{K}\,\mathsf{K}$. These are the only different rules;
the padding subtrees are identical and productive in both outputs.

\paragraph{Running example.}
Take $M=\mathsf{inc}_A;\mathsf{dec}_A(0)$, starting from zero.
After the increment, $A=1$, $m=1$, and $W=1$.
The zero choice advances to position $3$, while the positive choice
decrements $A$ and jumps to $0$. Both positions are out of range,
but only the positive choice is legal.
From \cref{eq:overview-polynomial}, $C_1(u,v)=(u,v+1)$.
The false zero choice therefore multiplies the initial weight pair
$(1,0)$ by $C_1(1,0)=(1,1)$ and obtains $(1,1)$.
The positive choice multiplies by the old counter pair $(1,0)$ and
retains weight $(1,0)$. The observations appear in
\cref{fig:reduction}.

\input{figures/reduction}

In the compiled execution, the increment constructs
$U=\mathsf{Succ}\,\mathsf{Zero}$. This term and $\mathsf{One}$ both
represent the number $1$.
Let $p_1,q_1$ be the updated polynomial terms, representing
$P_1^+,P_1^-$. The false zero branch passes the weight pair
$(h_0,k_0)$ in both schemes, where
\[
 h_0=\mathsf{Pos}\,\mathsf{One}\,\mathsf{Zero}\,
       (p_1\,U\,\mathsf{Zero})\,(q_1\,U\,\mathsf{Zero})
\]
and $k_0$ is the corresponding $\mathsf{Neg}$ term. Both components
represent one. At the halt, the observation terms $h_0\,a\,z$ and
$k_0\,a\,z$ both generate the chain $az$.
On the positive branch, the multiplier pair is $(U,\mathsf{Zero})$,
and the halt observations are $az,z$. Choosing child $1$ at the
increment, the corresponding observation addresses are $1\,1\,0$
and $1\,2\,0$; the latter already has different labels.

\paragraph{Representations on every candidate.}
The stored numerals must remain usable in later arithmetic, even when
they are composed terms such as $U$. Correctness only at the observation
arguments $a,z$ does not justify these uses. For example,
\[
 \mathsf{Mul}\,m\,n\,f\,x\to m\,(n\,f)\,x
\]
supplies $m$ with a function computed from $n$.
A numeral must therefore implement the represented number of iterations
for every related function and ground argument. If $m,n$ represent
$r,s$ and $f$ represents a function $F$, then $n\,f$ represents $F^s$;
applying $m$ gives $rs$ iterations of $F$.
The typed realization relation in \cref{eq:arithmetic-relation}
formalizes this requirement for numerals and polynomial functions.

Induction on candidate-prefix length maintains the numeral property
for all six natural-number components and the function property for
both polynomial terms. This includes illegal candidates with zero
signed weight or negative virtual counters: every component still
represents a finite natural number, as proved in \cref{lem:registers}.
The witness invariant identifies legal prefixes through the signed
weight; the representation invariant validates arithmetic on every candidate.

\paragraph{Productivity.}
A control rule exposes $b$ immediately, but its observation child may
still require arithmetic. Numeral realization ensures that the selected
component produces a finite chain. The continuation children again
satisfy the representation invariant, or enter common padding.
Induction on observation depth therefore proves productivity of both
whole trees, including every illegal branch. Guarded control recursion
alone would leave the arithmetic observations unjustified.

\subsection{Closing both directions of the reduction}

Let $v\in\{1,2\}^*$ encode a candidate prefix by its continuation
choices. Corresponding control nodes in the two schemes carry the
same register terms. The all-address correspondence of
\cref{lem:state} identifies their observations with those specified
by the arithmetic updates, including the padding and absent addresses.

If $v$ reaches a legal halt with weight components $r,s$, then
\cref{eq:overview-witness} gives $r\ne s$. Its observation children
are $a^rz$ and $a^sz$. At the address
\begin{equation}\label{eq:overview-disagreement}
                         v\,0\,0^{\min(r,s)},
\end{equation}
one tree has $z$ and the other has $a$. Thus a source halt gives an
actual finite HORS discrepancy.

For the converse, take a first discrepancy. Productivity and the
ranked alphabet ensure that both nodes there exist: their agreeing
ancestor labels determine the same child arities. Following only
children $1,2$ cannot give different labels, since every control or
padding root is $b$. Padding subtrees agree completely. A discrepancy
must therefore select child $0$ of a candidate state. A nonhalting
state has observation $z$ on both sides. At a halt the only possible
disagreement is between unequal chains, so $r\ne s$ and the signed
weight is nonzero. The entire candidate is legal by
\cref{eq:overview-witness}, and its control position is a source halt.
Write $T_\diamond=\BT(G_M^\diamond)$, with $\dom(T)$ denoting
a tree's address set and $T(w)$ its label at $w$. Then
\[
\begin{aligned}
 M\halts
 &\Longleftrightarrow \exists v\in\Cand_M\;
       [\pc_M(v)\notin I_M\land W_M(v)\ne0]\\
 &\Longleftrightarrow \exists w\in\dom(T_L)\cap\dom(T_R)\;
       [T_L(w)\ne T_R(w)]\\
 &\Longleftrightarrow \BT(G_M^L)\ne\BT(G_M^R).
\end{aligned}
\]
The intermediate compiler only inserts fixed finite templates into a
table indexed by the source instructions, so it is total computable.
\Cref{thm:compiler} states this stage's specification.

\subsection{Realizing the two-terminal alphabet}

The constructors $z,a,b$ separate the observation and control structure
throughout the argument. The final compiler pass realizes each by a
fixed tree context over $\Sigma_2=\{e:0,d:2\}$:
\[
\begin{aligned}
 E(z)&=e,\\
 E(a\,t)&=d(e,E(t)),\\
 E(b(t_0,t_1,t_2))
   &=d\bigl(d(e,e),d(E(t_0),d(E(t_1),E(t_2)))\bigr).
\end{aligned}
\]
The root separates $z$ from $a,b$, and the left child's root separates
$a$ from $b$. Their source children occur at fixed target addresses,
so decoding finite addresses recovers every source label.
Thus $E(T)=E(U)\iff T=U$, for finite and infinite trees.
Three helper rules implement these contexts at order at most one.
They preserve safety and productivity; \cref{lem:binary-encoding}
gives the syntactic and semantic argument.

For the final schemes $\widehat G_M^\diamond=(G_M^\diamond)^E$,
\[
 M\nonhalts\iff\BT(G_M^L)=\BT(G_M^R)
       \iff\BT(\widehat G_M^L)=\BT(\widehat G_M^R).
\]
This proves the lower bound of \cref{thm:main}, including its
two-terminal restriction. The constructor encoding is routine;
the witness implementation and its type and safety arguments are
unchanged by this final pass.

\subsection{The upper bound and completeness}

Since zero-start machine halting is c.e.-complete
(\cref{sec:preliminaries}), the reduction from its complement gives
$\PiOne$-hardness for the stated restricted class.
The general co-c.e.\ upper bound already follows from Jochems's effective
$\bot$-elimination~\cite[Lemma~2.7 and Proposition~2.8]{jochems2021}.
Its productive outputs permit finite-prefix search for a discrepancy.
\Cref{sec:upper} gives a direct argument using the standard finite
model for head convergence~\cite[Section~4.1 and Theorem~4.1]{salvati2015}.

For a closed ground term, this model records only whether a terminal
head can be exposed; $z$ and $az$ both have value one. Higher-type
values are finite function tables. Monotone iteration of the rule
interpretations from the least values stabilizes, and head adequacy
makes the final answer exact. To read an address, start at the root.
A positive answer permits a terminating head search and, if the address
continues, descent through a valid child index. A negative answer makes
the current node $\bot$; all its proper descendants are absent.
An out-of-range child index is also absent.

The resulting reader distinguishes a terminal label, a present
$\bot$ node, and an absent address, and computes an exact prefix
$\Pref(G,d)$ at any prescribed finite depth $d$.
For well-formed schemes over a common ranked alphabet,
\[
 \BT(G)=\BT(H)\quad\Longleftrightarrow\quad
             \forall d\;\Pref(G,d)=\Pref(H,d).
\]
Productivity means that no such prefix contains $\bot$.
Together with the decidable syntactic restrictions, these universal
conditions give a $\PiOne$ definition of the restricted problem as
an ordinary language. Combining this upper bound with the reduction
proves the completeness claim of \cref{thm:main}.

%% file: figures/reduction.tex
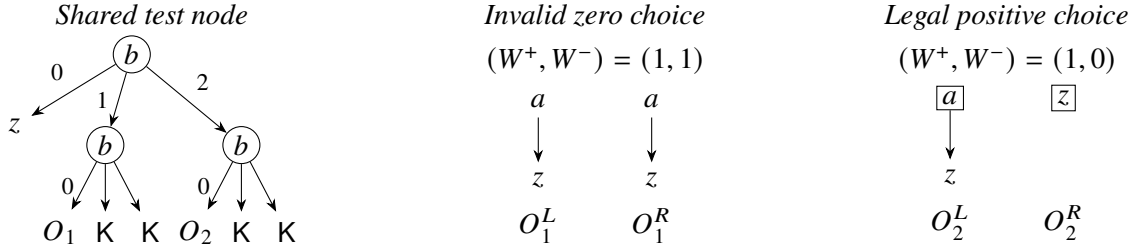
\begin{figure*}[htbp]
\centering
\begin{minipage}[t]{0.34\textwidth}
\centering
\textit{Shared test node}\par\smallskip
\begin{tikzpicture}[>=Stealth,every node/.style={font=\normalsize},
  state/.style={circle,draw,inner sep=1.5pt}]
 \node[state] (t) at (0,0) {$b$};
 \node (z) at (-1.55,-0.95) {$z$};
 \node[state] (h0) at (-0.35,-1.2) {$b$};
 \node[state] (h1) at (1.45,-1.2) {$b$};
 \draw[->] (t)--node[above left,font=\footnotesize]{$0$}(z);
 \draw[->] (t)--node[left,font=\footnotesize]{$1$}(h0);
 \draw[->] (t)--node[above right,font=\footnotesize]{$2$}(h1);
 \foreach \h/\x/\o in {h0/-0.35/1,h1/1.45/2}{
   \node (o\o) at (\x-0.6,-2.35) {$O_\o$};
   \node (k\o) at (\x,-2.35) {$\mathsf{K}$};
   \node (j\o) at (\x+0.6,-2.35) {$\mathsf{K}$};
   \draw[->] (\h)--node[left,font=\footnotesize]{$0$}(o\o);
   \draw[->] (\h)--(k\o);
   \draw[->] (\h)--(j\o);
 }
\end{tikzpicture}
\end{minipage}\hfill
\begin{minipage}[t]{0.29\textwidth}
\centering
\textit{Invalid zero choice}\par\smallskip
$(W^+,W^-)=(1,1)$\par\smallskip
\begin{tikzpicture}[>=Stealth,every node/.style={font=\normalsize}]
 \node (l) at (0,0) {$a$}; \node (r) at (1.5,0) {$a$};
 \node (lz) at (0,-1.05) {$z$}; \node (rz) at (1.5,-1.05) {$z$};
 \draw[->] (l)--(lz); \draw[->] (r)--(rz);
 \node at (0,-1.65) {$O_1^L$}; \node at (1.5,-1.65) {$O_1^R$};
\end{tikzpicture}
\end{minipage}\hfill
\begin{minipage}[t]{0.29\textwidth}
\centering
\textit{Legal positive choice}\par\smallskip
$(W^+,W^-)=(1,0)$\par\smallskip
\begin{tikzpicture}[>=Stealth,every node/.style={font=\normalsize}]
 \node[draw,inner sep=2pt] (l) at (0,0) {$a$};
 \node[draw,inner sep=2pt] (r) at (1.5,0) {$z$};
 \node (lz) at (0,-1.05) {$z$};
 \draw[->] (l)--(lz);
 \node at (0,-1.65) {$O_2^L$}; \node at (1.5,-1.65) {$O_2^R$};
\end{tikzpicture}
\end{minipage}
\caption{The program $\mathsf{inc}_A;\mathsf{dec}_A(0)$ at its test
after the increment.
Both schemes generate this whole control skeleton; $\mathsf{K}$ marks
the common padding subtrees. The observations $O_1,O_2$ correspond to
the zero and positive choices, respectively. The invalid zero choice gives weight components $(1,1)$ and
matching observations $az$; the legal positive choice gives $(1,0)$
and observations $az,z$. The differing labels are boxed.}
\label{fig:reduction}
\Description{A shared ternary test node branches to two halt nodes.
For the invalid zero branch both schemes output an a node above a z
leaf. For the legal positive branch the left observation has an a
node above z, while the right observation is a z leaf.}
\end{figure*}

%% file: sections/02-preliminaries.tex
\section{Schemes, observations, and decision problems}\label{sec:preliminaries}

Simple types are generated by $o$ and $\sigma\to\tau$, with arrows
associating to the right; $\sigma^r\to\tau$ denotes $r$ curried
arguments of type $\sigma$. We use the order function defined in
\cref{eq:type-order}.
A ranked alphabet assigns a natural-number arity $r$ to each terminal
$a$, whose type is $o^r\to o$. A deterministic higher-order recursion
scheme $G$ consists of a finite such alphabet $\Sigma$, a finite typed
set of nonterminals, a start symbol $S:o$, and exactly one rule
$F x_1\cdots x_r\to t_F$ per nonterminal
$F:\sigma_1\to\cdots\to\sigma_r\to o$.
The parameters are distinct, $x_i:\sigma_i$, and $t_F:o$ is an
applicative term over these parameters, terminals, and declared
nonterminals. Terms have the syntax $t::=x\mid F\mid a\mid t\,t$.
Thus a rule has a ground body even when a partial application of its
nonterminal has functional type. The scheme order is the maximum
order of its nonterminal types.
Concrete nonterminal names in the constructions, such as
$\mathsf{F}_i$, $\mathsf{S}$, and $\mathsf{K}$, are set in sans serif;
schematic symbols and rule parameters are italic.

\paragraph{Safety.}
We use \cite[Definition~1]{hague2017}. A type
$\sigma_1\to\cdots\to\sigma_r\to o$ is homogeneous if each
$\sigma_i$ is homogeneous and
$\ord(\sigma_1)\ge\cdots\ge\ord(\sigma_r)$.
A rule is safe if its nonterminal and all body subterms have
homogeneous types, and every application operand $s$ satisfies
$\ord(s)\le\ord(x)$ for each rule parameter $x$ occurring in $s$.
In $u\,s$, this condition applies to $s$, not to the function prefix
$u$. Nonterminal constants do not count as parameters.
A scheme is safe if all its rules are safe. Typing and finite syntax
traversal decide this property.

\paragraph{Head rewriting.}
For closed ground terms, the only unfolding rule is
\[
 F\,u_1\cdots u_r\;\longrightarrow_h\;
 t_F[u_1/x_1,\ldots,u_r/x_r],
\]
where substitution is simultaneous. No reduction under a terminal is
needed to expose its head. Write $t\Downarrow_G a(t_0,\ldots,t_{r-1})$
when finitely many head steps produce that terminal application.
A well-typed closed ground term has either a terminal or a fully applied
nonterminal at its head. Unfolding preserves typing, by induction on
the body with simultaneous substitution. Since each nonterminal has one
rule, head steps are deterministic, and two head exposures of the same
term have the same terminal and the same list of children.

Addresses are words over $\N$, with zero-based child indices. Define
positive observations by
\[
\begin{aligned}
 \Obs_G(t,\epsilon,a)&\iff \exists\vec u\;t\Downarrow_G a(\vec u),\\
 \Obs_G(t,iw,a)&\iff
 \exists c,\vec u\;[t\Downarrow_G c(\vec u)\land i<|\vec u|
                         \land\Obs_G(u_i,w,a)].
\end{aligned}
\]
The value tree $\BT_G(t)$ has root $a$ and recursively these children
when $t\Downarrow_G a(\vec u)$; otherwise its root is a fresh nullary
symbol $\bot$. This defines a possibly infinite ranked tree over
$\Sigma\cup\{\bot:0\}$. In particular, a child below a leaf is absent,
whereas a divergent term at an actual child gives a present $\bot$ node.
The positive observations determine this tree: the root either has its
unique observed label or is $\bot$, and the same reasoning applies to
each child required by an observed label's arity. Hence
\begin{equation}\label{eq:observation-equality}
 \BT(G)=\BT(H)\iff
 \forall w\in\N^*,\ a\in\Sigma\;
 [\Obs_G(S_G,w,a)\leftrightarrow\Obs_H(S_H,w,a)].
\end{equation}
Here and below $\BT(G)=\BT_G(S_G)$, with a shared ranked alphabet.

A term is productive to depth zero unconditionally; it is productive
to depth $d+1$ if it exposes a terminal and every actual child is
productive to depth $d$. A scheme is \emph{productive} if its start
term is productive to every depth. This is equivalent to having no
$\bot$ node; finite terminal leaves are permitted.

\paragraph{Inputs.}
Source programs and HORS descriptions are finite syntax. Fix standard
computable injective encodings by natural numbers, with total computable decoders
that return syntax or reject the input.
Validity, closed references, and a specified order bound are decidable
by syntax-directed checks. We consider ordinary languages of encodings:
\[
\begin{aligned}
 \EqH&=\{\langle G,H\rangle:\Valid(G,H)\land\BT(G)=\BT(H)\},\\
 \EqFour&=\{\langle G,H\rangle\in\EqH:
                   \Sigma=\Sigma_2\land\max\{\ord(G),\ord(H)\}\le4\},\\
 \EqProd&=\{\langle G,H\rangle\in\EqFour:
                                  \Prod(G)\land\Prod(H)\},\\
 \EqSafeProd&=\{\langle G,H\rangle\in\EqProd:
                                  \Safe(G)\land\Safe(H)\}.
\end{aligned}
\]
Malformed inputs are rejected. Membership in $\EqProd$ includes
productivity; the input supplies no semantic promise.

\paragraph{Source machines.}
We use the two-counter machine model with undecidable halting
in~\cite[Definition~2 and Theorem~6, pp.~16:3--16:4]{dudenhefner2022}, with instruction
positions reindexed from zero to one. \Cref{sec:overview-candidates}
gives the instruction syntax and execution convention.
A configuration $(q,a,b)\in\N^3$ records the program position and
the two counter values. It halts exactly when $q\notin I_M$.
Write $\to_M$ for the configuration step relation and $\to_M^*$
for its reflexive transitive closure. Our termination notation is
\[
\begin{aligned}
 M\halts&\iff\exists(q,a,b)\in\N^3\;
       [(1,0,0)\to_M^*(q,a,b)\land q\notin I_M],\\
 M\nonhalts&\iff\neg(M\halts).
\end{aligned}
\]
The zero-start halting problem is c.e.-complete.
The effective two-counter simulation gives c.e.-hardness for arbitrary
initial configurations~\cite[Section~2, p.~16:4]{dudenhefner2022}.
A finite prefix initializes the counters, then uses the two-instruction
gadget $\mathsf{inc}_A;\mathsf{dec}_A(j)$ to jump to the relocated
initial position $j$ while restoring both counter values: the decrement
necessarily takes its positive branch. Jump relocation keeps every
halting target out of range, in particular leaving target zero at zero.
This yields the fixed initial state $(1,0,0)$.
Bounded simulation gives c.e.\ membership.

%% file: sections/03-witness.tex
\section{Arithmetic witnesses for candidate executions}\label{sec:witness}

We adapt the bounded-counter guard witness of
Boiret, Pi\'orkowski, and Schmude~\cite[Section~4, pp.~48:12--48:14]{boiret2018}
to the source instruction set in \cref{sec:overview-candidates}.
We first define the candidate unfolding, then assign weights that
distinguish legal prefixes from illegal ones.

\subsection{Candidate unfolding}\label{sec:candidate-semantics}

A candidate follows the program's control instructions but chooses
either outcome of a decrement test without checking its guard.
Its virtual counters range over $\Z$, since a wrong positive choice
at zero produces $-1$. For an active configuration $s=(i,a,b)$,
where $i\in I_M$, and a choice $e\in\{1,2\}$, define
\begin{equation}\label{eq:candidate-successor}
 \delta_e^M(i,a,b)=
 \begin{cases}
  (i+1,a+1,b),&M(i)=\mathsf{inc}_A,\\
  (i+1,a,b+1),&M(i)=\mathsf{inc}_B,\\
  (i+1,a,b),&M(i)=\mathsf{dec}_A(j)\land e=1,\\
  (j,a-1,b),&M(i)=\mathsf{dec}_A(j)\land e=2,\\
  (i+1,a,b),&M(i)=\mathsf{dec}_B(j)\land e=1,\\
  (j,a,b-1),&M(i)=\mathsf{dec}_B(j)\land e=2.
 \end{cases}
\end{equation}
Both choices have the same successor at an increment.
The map is undefined when $i\notin I_M$.
For $M(i)$ acting on counter $X\in\{A,B\}$, let $c$ be that
counter's value in $s$. Its branch condition is
\[
 \gamma_e^M(s)\iff
 \begin{cases}
  \top,&M(i)=\mathsf{inc}_X,\\
  c=0,&M(i)=\mathsf{dec}_X(j)\land e=1,\\
  c>0,&M(i)=\mathsf{dec}_X(j)\land e=2.
 \end{cases}
\]
The successor $\delta_e^M(s)$ is defined independently of
$\gamma_e^M(s)$. On natural-counter configurations, the choices
satisfying $\gamma_e^M(s)$ give exactly the source steps $\to_M$.

The candidate tree has domain $\Cand_M\subseteq\{1,2\}^*$ and
configuration map $s_M:\Cand_M\to\N\times\Z^2$, generated by
\[
\begin{gathered}
 \epsilon\in\Cand_M,\qquad s_M(\epsilon)=(1,0,0),\\[1ex]
 \frac{v\in\Cand_M\quad\pc_M(v)\in I_M\quad e\in\{1,2\}}
      {ve\in\Cand_M\quad s_M(ve)=\delta_e^M(s_M(v))}.
\end{gathered}
\]
Here $\pc_M(v)$ is the first coordinate of $s_M(v)$, and
$\Cand_M$ is the least set closed under these clauses.
Every active node has two children; a node with
$\pc_M(v)\notin I_M$ is a leaf. Thus a finite program may have
an infinite candidate unfolding.
The legal prefixes form $\Legal_M\subseteq\Cand_M$, defined by
\[
\begin{aligned}
 \epsilon&\in\Legal_M,\\
 ve\in\Legal_M&\iff v\in\Legal_M\land\gamma_e^M(s_M(v)),
                         \qquad ve\in\Cand_M.
\end{aligned}
\]
They follow the source execution, with either label available at
increments. The output trees in \cref{eq:model-tree} add observation
children and common halt padding to this candidate tree.

\subsection{Guard weights and halting correspondence}

After $m$ instructions let
\[
 Z_m(X)=\prod_{j=1}^{m}(X-j),\qquad Z_0(X)=1.
\]
A legal prefix has $0\le A,B\le m$, because an instruction increases
each counter by at most one. Also $Z_m(0)=(-1)^m m!\ne0$, while
$Z_m(c)=0$ for $1\le c\le m$.
Each weight update multiplies by a factor determined by the current
instruction $\iota$ on counter $X$, the candidate choice $e$, and
the old counter value $c$:
\begin{equation}\label{eq:weight-factor}
 \phi_m(\iota,e,c)=
 \begin{cases}
  1,&\iota=\mathsf{inc}_X,\\
  Z_m(c),&\iota=\mathsf{dec}_X(j)\land e=1,\\
  c,&\iota=\mathsf{dec}_X(j)\land e=2.
 \end{cases}
\end{equation}
For $ve\in\Cand_M$, take $m=|v|$, $\iota=M(\pc_M(v))$, and
$c$ as the value of its selected counter after $v$. Define
\begin{equation}\label{eq:weight-recurrence}
 W_M(\epsilon)=1,\qquad
 W_M(ve)=W_M(v)\,\phi_m(\iota,e,c).
\end{equation}
Counter and control updates are given by \cref{eq:candidate-successor};
the unselected counter is unchanged.
Every instruction advances the clock to $m+1$ and the polynomial to
$Z_{m+1}$. The multiplier in \cref{eq:weight-factor} uses the
pre-transition counter and polynomial, including before a decrement.

In the register construction of~\cite{boiret2018}, an input word can
name arbitrary transitions. Its weight update is
$r_w\leftarrow r_wT_qT_1T_2$: $T_q$ checks the source control state,
and $T_1,T_2$ check both counter guards. A further state factor in the
output selects the target state. Here candidates follow the finite
control by construction; increments have no guard, and a decrement
tests one counter. We retain exactly that guard's multiplier and
observe weights only at halts. The complete numerical weights therefore
need not coincide with those of~\cite{boiret2018}; the shared ingredients
are $Z_m$, its guard factors, and the following first-error proof.

\begin{lemma}[Witness invariant]\label{lem:witness}
\[
 \forall v\in\Cand_M,\qquad
 W_M(v)\ne0\quad\Longleftrightarrow\quad v\in\Legal_M.
\]
\end{lemma}
\begin{proof}
On a legal prefix, increments preserve the weight, positive choices
multiply it by a positive integer, and zero choices multiply it by
$Z_m(0)\ne0$. Thus all legal prefixes have nonzero weight.
If a prefix is illegal, consider its first erroneous choice. All earlier
choices were legal, so the current tested counter lies in $[0,m]$.
A false zero choice has $1\le c\le m$ and contributes $Z_m(c)=0$.
A false positive choice has $c=0$ and contributes zero. Subsequent
updates only multiply the weight or leave it unchanged, so it remains
zero, regardless of any later negative virtual counters.
\end{proof}

\begin{lemma}[Control correspondence]\label{lem:live-halt}
\[
 M\halts\quad\Longleftrightarrow\quad
 \exists v\in\Cand_M\;[\pc_M(v)\notin I_M\land W_M(v)\ne0].
\]
\end{lemma}
\begin{proof}
Given a machine execution, choose its actual guard at each test and
either label at an increment. Induction on its length identifies the
candidate program counter and virtual counters with the source
configuration; \cref{lem:witness} gives nonzero weight. A source halt
therefore supplies the required candidate.
Conversely, nonzero weight makes every choice legal. Starting from
natural counters zero, induction on the candidate steps shows that an
increment is the source increment, a zero choice is its zero case,
and a positive choice subtracts one from a positive natural counter
and takes its jump. The final out-of-range control position has no
source successor, so this candidate is a halting source execution.
\end{proof}

%% file: sections/04-compiler.tex
\section{The compiler and its correctness}\label{sec:compiler}

The first compiler stage inserts fixed rule templates into a table
indexed by the source instructions. These rules implement the arithmetic
and control updates over $\Sigma_0$ described in \cref{sec:overview}.
The final pass in \cref{sec:binary-target} realizes the constructors
over $\Sigma_2$.
Every rule has an applicative ground body. Tuples, primed terms, and
side selection below are textual abbreviations expanded before the
rules are emitted.

\subsection{Arithmetic rules and polynomial updates}\label{sec:rules}

For a natural pair $C=(u,v)$ write $\val(C)=u-v\in\Z$, and put
\[
 \begin{aligned}
 \operatorname{inc}(u,v)&=(u+1,v),&
 \operatorname{dec}(u,v)&=(u,v+1),\\
 (u,v)\otimes(r,s)&=(ur+vs,us+vr).&&
 \end{aligned}
\]
Thus $\val(\operatorname{inc}C)=\val(C)+1$,
$\val(\operatorname{dec}C)=\val(C)-1$, and $\val(C\otimes D)=\val(C)\,\val(D)$.
All components remain natural numbers, including on illegal candidates.

Set $\tyC=o\to o$, $\NN=\tyC\to\tyC$, and
$\PP=\NN\to\NN\to\NN$. The numerical rules are
\begin{align*}
 \mathsf{Zero}\ f\ x&\to x,\\
 \mathsf{One}\ f\ x&\to f\,x,\\
 \mathsf{Succ}\ n\ f\ x&\to f\,(n\,f\,x),\\
 \mathsf{Add}\ m\ n\ f\ x&\to m\,f\,(n\,f\,x),\\
 \mathsf{Mul}\ m\ n\ f\ x&\to m\,(n\,f)\,x,\\
 \mathsf{Pos}\ u\ v\ r\ s\ f\ x
    &\to\mathsf{Add}\,(\mathsf{Mul}\ u\ r)\,
                       (\mathsf{Mul}\ v\ s)\,f\,x,\\
 \mathsf{Neg}\ u\ v\ r\ s\ f\ x
    &\to\mathsf{Add}\,(\mathsf{Mul}\ u\ s)\,
                       (\mathsf{Mul}\ v\ r)\,f\,x .
\end{align*}
Here $f:\tyC$ and $x:o$; every right-hand side has type $o$.
Juxtaposition is left-associative.

Define a pair of natural-valued functions $C_m=(P_m^+,P_m^-)$ by
\begin{equation}\label{eq:pair-poly}
 C_0(u,v)=(1,0),\qquad
 C_{m+1}(u,v)=(u,v+1)\otimes C_m(u,v+1).
\end{equation}
Induction on $m$ gives
\begin{equation}\label{eq:pair-poly-value}
 \forall u,v\in\N,\qquad P_m^+(u,v)-P_m^-(u,v)=Z_m(u-v).
\end{equation}
The successor difference is
$(u-v-1)Z_m(u-v-1)=Z_{m+1}(u-v)$, using
\begin{equation}\label{eq:poly-shift}
 Z_{m+1}(x)=(x-1)Z_m(x-1).
\end{equation}
This identity holds for all integers $x$, not only legal counter values.

The seed nonterminals have type $\PP$ and rules
$\mathsf{Seed}_+\ u\ v\ f\ x\to\mathsf{One}\ f\ x$ and
$\mathsf{Seed}_-\ u\ v\ f\ x\to\mathsf{Zero}\ f\ x$.
The update nonterminals have type $\PP\to\PP\to\PP$:
\begin{equation}\label{eq:next-rules}
\begin{aligned}
 \mathsf{Next}_+\ p\ q\ u\ v\ f\ x
   &\to \mathsf{Pos}\ u\ (\mathsf{Succ}\ v)\\[-2pt]
   &\hspace{1.5em}(p\ u\ (\mathsf{Succ}\ v))\,
                  (q\ u\ (\mathsf{Succ}\ v))\ f\ x,\\
 \mathsf{Next}_-\ p\ q\ u\ v\ f\ x
   &\to \mathsf{Neg}\ u\ (\mathsf{Succ}\ v)\\[-2pt]
   &\hspace{1.5em}(p\ u\ (\mathsf{Succ}\ v))\,
                  (q\ u\ (\mathsf{Succ}\ v))\ f\ x .
\end{aligned}
\end{equation}
Here $p,q:\PP$, $u,v:\NN$, $f:\tyC$, and $x:o$.
The right-hand sides have type $o$. Each occurrence of $p$ or $q$ is
an ordinary application to two numeral terms.

\subsection{Compiling the control flow}\label{sec:control-rules}

Fix a side $\diamond\in\{L,R\}$.
\Cref{fig:control-compiler} gives a control-rule template for each
source instruction. At position $i$, $M(i)$ selects the rule defining
$\mathsf{F}_i$. These selected rules form $\mathcal R_{\mathrm{ctrl}}(M)$.
The program translation $\llbracket M\rrbracket^\diamond$ assembles
them with the arithmetic library of \cref{sec:rules} and the
halt, padding, and start rules to obtain $G_M^\diamond$.
All abbreviations are expanded before the rules are emitted.

\input{figures/control-compiler}

Each source position contributes one rule. The common root $b$ has
observation child $z$ and continuation
children $1,2$. An increment duplicates its successor call. A decrement
test generates both guesses: child $1$ keeps the selected counter and
advances to the next instruction; child $2$ increments its negative
component and jumps to the specified position.

\Needspace{6\baselineskip}
Every continuation receives $p',q'$, so the polynomial advances on
every instruction. The zero guess multiplies the weight by the
\emph{old} polynomial evaluated at the \emph{old} counter pair.
The positive guess uses the \emph{old} counter pair as its multiplier,
before the negative component is incremented. These simultaneous
updates are explicit in the four control-rule templates.

Normalization sends every out-of-range position to $0$, represented
by $\mathsf{F}_0$. The choice $w_L=h$, $w_R=k$ is fixed by the output side:
the two halt observations are $h\,a\,z$ and $k\,a\,z$, and both use
$\mathsf{K}$ for their continuation children. Only this halt rule differs
between the two schemes. The start rule supplies the polynomial seeds,
zero counter components, and weight components $(\mathsf{One},\mathsf{Zero})$.

\PaperFloatBarrier

\subsection{Well-formedness, order, and safety}\label{sec:safety-audit}

The construction contains eleven arithmetic nonterminals, $\mathsf{S},\mathsf{K}$,
and the $L+1$ control nonterminals. The source instruction list selects
one control template per state; no source counter or execution is
evaluated during compilation.

\begin{lemma}[Syntactic compilation]\label{lem:typing}
The construction produces two finite, closed, well-typed deterministic
HORS over $\Sigma_0$, of order at most four.
\end{lemma}
\begin{proof}
There are $L+1$ states and thirteen other nonterminals, each with one
rule. All referenced state indices have been normalized into $0,\ldots,L$.
The type audit is
\[
\begin{array}{lll}
 \mathsf{Zero},\mathsf{One}:\NN & &2\\
 \mathsf{Succ}:\NN\to\NN &&3\\
 \mathsf{Add},\mathsf{Mul}:\NN\to\NN\to\NN &&3\\
 \mathsf{Pos},\mathsf{Neg}:\NN^4\to\NN &&3\\
 \mathsf{Seed}_+,\mathsf{Seed}_-:\PP &&3\\
 \mathsf{Next}_+,\mathsf{Next}_-:\PP\to\PP\to\PP &&4\\
 \mathsf{F}_i:\PP\to\PP\to\NN^6\to o &&4\\
 \mathsf{S},\mathsf{K}:o &&0 .
\end{array}
\]
The last column is the order of the \emph{whole} nonterminal type.
For example, in \cref{eq:next-rules}, $\mathsf{Succ}\ v:\NN$,
$p\,u\,(\mathsf{Succ}\ v):\NN$, and the saturated $\mathsf{Pos}$
or $\mathsf{Neg}$ application has type $o$.
In the instruction clauses of \cref{fig:control-compiler},
$p',q':\PP$ and every remaining
state argument has type $\NN$, so each continuation has type $o$.
In its halt clause, $h\,a\,z$ and $k\,a\,z$ have type $o$.
The numeral parameter context for $\mathsf{Add}$ and $\mathsf{Mul}$
is $m:\NN,n:\NN,f:\tyC,x:o$.
For addition, $n\,f\,x:o$ and $m\,f:\tyC$, giving a ground body.
For multiplication, $n\,f:\tyC$, $m\,(n\,f):\tyC$, and application
to $x$ again yields $o$.
For $\mathsf{Pos}$ and $\mathsf{Neg}$, each partial multiplication
has type $\NN$; their sum has type $\NN$ before the final $f,x$
applications. The seed rules ignore their first two $\NN$ arguments
and apply a numeral of type $\NN$ to $f,x$.
For a polynomial update, the parameter context has two $\PP$
arguments followed by two $\NN$ arguments, $f$, and $x$.
Its two polynomial calls have type $\NN$, so the component
multiplication rule gives the required ground result.
\end{proof}

\begin{lemma}[Safety of the compiler]\label{lem:safety}
For every $M$ and side $\diamond$, the compiled scheme $G_M^\diamond$
is safe.
\end{lemma}
\begin{proof}
We use the syntactic condition of~\cite[Definition~1]{hague2017}.
The types $o,\tyC,\NN,\PP$ are homogeneous, of orders $0,1,2,3$.
Every arithmetic type in the preceding audit expands to
$\PP^r\to\NN^s\to\tyC\to o\to o$ for some $r,s\in\N$.
Its parameter orders are therefore nonincreasing. The same holds for
the state type $\PP^2\to\NN^6\to o$, the ground types of $\mathsf{S},\mathsf{K}$,
and terminal types with only ground arguments.
Every subterm type is a codomain
suffix of the type of its head symbol, which is a parameter,
nonterminal, or terminal. Those suffixes remain homogeneous.

The table lists every family of compound nonground operands after
expansion of the rule abbreviations, up to renaming parameters of the
same type. Here $p,q:\PP$, $u,v,h,k,n:\NN$, and $f:\tyC$.
Each $\mathsf{Pos}$ row also covers the corresponding $\mathsf{Neg}$
term. The last column gives the set of orders of parameters occurring
in the operand.
\begin{center}
\begin{tabular}{@{}lcc@{}}
\toprule
Operand $s$ & $\ord(s)$ & Parameter orders\\
\midrule
$\mathsf{Next}_\pm\,p\,q$ & $3$ & $\{3\}$\\
$\mathsf{Succ}\,u$, $\mathsf{Mul}\,u\,v$ & $2$ & $\{2\}$\\
$p\,u\,v$, $q\,u\,v$ & $2$ & $\{2,3\}$\\
$p\,u\,(\mathsf{Succ}\,v)$, $q\,u\,(\mathsf{Succ}\,v)$ & $2$ & $\{2,3\}$\\
$\mathsf{Pos}\,h\,k\,u\,v$ & $2$ & $\{2\}$\\
$\mathsf{Pos}\,h\,k\,(p\,u\,v)\,(q\,u\,v)$ & $2$ & $\{2,3\}$\\
$n\,f$ & $1$ & $\{1,2\}$\\
\bottomrule
\end{tabular}
\end{center}
The same numeric families cover the $B$-counter parameters $r,s$.
Ground operands impose only a lower bound of zero. A single parameter
has exactly its own order, and a closed operand has no parameters.
The start rule uses closed seeds and numerals. Halt observations and
padding add only ground operands, the closed terminal $a$, and $\mathsf{K}$.
This checks the condition at every application node in every rule.

For example, in $p\,u\,(\mathsf{Succ}\ v)$ the prefix $p\,u$ has
order three and contains the order-two parameter $u$. It is the
function child of the outer application, so it has no operand
inequality of its own. The whole polynomial application has order
two, and its uses as an operand are safe. Applying the inequality
to every application prefix would impose a different restriction.
\end{proof}

\subsection{A representation invariant that supports application}

Multiplication reuses a numeral at a functional argument of the form
$n\,f$. Correctness only at the observation arguments $a,z$ would leave
this application unproved. We therefore require correctness at every
represented function and ground argument.

Fix one compiled scheme $G$. Let $V_o$ be the set of partial observation
maps $T:\N^*\rightharpoonup\Sigma_0$, writing $T(w)=\dagger$ for
$w\notin\dom(T)$. These maps record terminal observations only.
We allow arbitrary partial maps here; \cref{sec:operational-correspondence}
shows that the constructed output maps satisfy the rank conditions and
define productive trees. Set
$V_{\sigma\to\tau}=V_\tau^{V_\sigma}$.
For a closed term $t:\sigma$ and $v\in V_\sigma$, define a logical
relation by
\begin{equation}\label{eq:arithmetic-relation}
\begin{aligned}
 t\Vdash_o T
   &\iff \forall w,a\;[\Obs_G(t,w,a)\leftrightarrow T(w)=a],\\
 t\Vdash_{\sigma\to\tau}F
   &\iff \forall u,U\;[u\Vdash_\sigma U
                          \ \Longrightarrow\ t\,u\Vdash_\tau F(U)].
\end{aligned}
\end{equation}
In the second line $u$ ranges over closed terms of type $\sigma$.
Ground realization is preserved when a head unfolding is
performed in reverse: if $t\to_h t'$ and $t'\Vdash_o T$, then
$t\Vdash_o T$, by determinism and the definition of observations.

For any closed $t:\NN$ and $n\in\N$, define $\Num(t,n)$ to mean
$t\Vdash_\NN(\lambda f\,X.f^n(X))$.
The term $t$ may contain nonterminals and partial applications, such as
$\mathsf{Succ}\,\mathsf{Zero}$; its numerical behavior is specified
under further application. The observation $t\,a\,z$ is one instance
of this requirement.

\begin{lemma}[Church arithmetic]\label{lem:numerals}
$\Num(\mathsf{Zero},0)$ and $\Num(\mathsf{One},1)$ hold.
If $\Num(t,m)$ and $\Num(u,n)$, then
\[
 \Num(\mathsf{Succ}\ t,m+1),\quad
 \Num(\mathsf{Add}\ t\ u,m+n),\quad
 \Num(\mathsf{Mul}\ t\ u,mn).
\]
For closed terms $t_1,t_2,t_3,t_4:\NN$ and
$\alpha,\beta,\gamma,\delta\in\N$, if
\[
 \Num(t_1,\alpha)\land\Num(t_2,\beta)\land
 \Num(t_3,\gamma)\land\Num(t_4,\delta),
\]
then
\[
\begin{aligned}
 \Num(\mathsf{Pos}\,t_1\,t_2\,t_3\,t_4,
                   \alpha\gamma+\beta\delta),\\
 \Num(\mathsf{Neg}\,t_1\,t_2\,t_3\,t_4,
                   \alpha\delta+\beta\gamma).
\end{aligned}
\]
The two represented numbers are the coordinates of the signed-pair
product $(\alpha,\beta)\otimes(\gamma,\delta)$.
Moreover, $\Num(t,n)$ implies $t\,a\,z\Vdash_o a^nz$.
\end{lemma}
\begin{proof}
Use the arithmetic rules of \cref{sec:rules} with arbitrary related
arguments $f\Vdash_{\tyC}F$ and $x\Vdash_o X$.
The zero and one rules return $x$ and $f\,x$; the successor rule
returns $f(t\,f\,x)$, realizing $F^{m+1}(X)$.
For addition the body $t\,f\,(u\,f\,x)$ realizes
$F^m(F^n(X))=F^{m+n}(X)$.
For multiplication, \cref{eq:arithmetic-relation} first gives
$u\,f\Vdash_{\tyC}F^n$. Applying the realization of $t$ to this
functional argument and then to $x$ shows that
$t\,(u\,f)\,x$ realizes $(F^n)^m(X)=F^{mn}(X)$.
One head unfolding of each saturated arithmetic call gives the
corresponding conclusion. For $\mathsf{Pos}$, multiplication gives
numerals for $\alpha\gamma$ and $\beta\delta$, and addition gives
$\alpha\gamma+\beta\delta$. The $\mathsf{Neg}$ rule similarly uses
$\alpha\delta$ and $\beta\gamma$.
Finally, the terminal $a$ realizes the unary constructor on observation
maps and $z$ realizes the one-node tree, giving the observation claim;
\cref{eq:constructor-maps} below defines these constructor maps explicitly.
\end{proof}

To verify the polynomial update, we relate the stored HORS terms to
the mathematical functions $P_m^+,P_m^-$.
For closed $p,q:\PP$, the representation predicate $\Pol_m(p,q)$
requires that their applications to any numeral representations of
$u,v$ represent the corresponding function values.
Here $t,s$ range over closed $\NN$-terms:
\begin{equation}\label{eq:polynomial-realization}
\begin{aligned}
 \Pol_m(p,q)\iff {}&\forall u,v\in\N\;\forall t,s:\NN,\\
 &\bigl(\Num(t,u)\land\Num(s,v)\bigr)\Longrightarrow\\
 &\quad\bigl(\Num(p\,t\,s,P_m^+(u,v))
                 \land\Num(q\,t\,s,P_m^-(u,v))\bigr).
\end{aligned}
\end{equation}
The index $m$ specifies the polynomial after $m$ candidate steps.
The next lemma shows that the two $\mathsf{Next}$ applications preserve
this relation at step $m+1$.
\begin{lemma}[Polynomial realization]\label{lem:polynomial}
The seeds satisfy $\Pol_0(\mathsf{Seed}_+,\mathsf{Seed}_-)$.
The update property is
\[
 \Pol_m(p,q)\ \Longrightarrow\
 \Pol_{m+1}(\mathsf{Next}_+\,p\,q,\mathsf{Next}_-\,p\,q).
\]
\end{lemma}
\begin{proof}
The seed rules reduce to one and zero on their last two arguments.
For the step, let $t,s$ represent $u,v$.
By \cref{lem:numerals}, $\mathsf{Succ}\ s$ represents $v+1$,
so the hypothesis applies to $p\,t\,(\mathsf{Succ}\ s)$ and
$q\,t\,(\mathsf{Succ}\ s)$. The two rules in \cref{eq:next-rules}
and the $\mathsf{Pos}/\mathsf{Neg}$ clauses of \cref{lem:numerals}
therefore give exactly the two components of
$(u,v+1)\otimes C_m(u,v+1)$ in \cref{eq:pair-poly}.
This argument applies to all $u,v$, including $u<v$.
\end{proof}

We combine the polynomial and numeral relations into an invariant for
a complete state call. Let $\vec t=(p,q,u,v,r,s,h,k)$ be its closed
argument tuple, with $p,q:\PP$ and $u,v,r,s,h,k:\NN$.
For natural-number pairs
$A=(A^+,A^-)$, $B=(B^+,B^-)$, and $W=(W^+,W^-)$, define
\begin{equation}\label{eq:register-realization}
\begin{aligned}
 \Reg_m(\vec t;A,B,W)\iff{}&\Pol_m(p,q)\\
 &\land\Num(u,A^+)\land\Num(v,A^-)\\
 &\land\Num(r,B^+)\land\Num(s,B^-)\\
 &\land\Num(h,W^+)\land\Num(k,W^-).
\end{aligned}
\end{equation}
This relates every stored component to its mathematical value;
the signed counter and weight values are $\val(A),\val(B),\val(W)$.
The initial terms in \cref{fig:control-compiler} represent
$A=B=(0,0)$ and $W=(1,0)$.
For the current instruction $\iota$ on counter $X$ and choice $e$,
update the selected counter pair $C$ and the weight pair simultaneously:
\begin{equation}\label{eq:pair-update}
 (C',W')=
 \begin{cases}
  (\operatorname{inc}C,W),&\iota=\mathsf{inc}_X,\\
  (C,W\otimes C_m(C)),&\iota=\mathsf{dec}_X(j)\land e=1,\\
  (\operatorname{dec}C,W\otimes C),&\iota=\mathsf{dec}_X(j)\land e=2.
 \end{cases}
\end{equation}
Here $C_m(C)=C_m(u,v)$ for $C=(u,v)$, and the other counter is
unchanged. Taking signed values gives counter changes $+1,0,-1$,
respectively, and, by \cref{eq:pair-poly-value},
\begin{equation}\label{eq:pair-weight-value}
                   \val(W')=\val(W)\,\phi_m(\iota,e,\val(C)).
\end{equation}
Thus these natural-pair operations realize the integer updates of
\cref{eq:overview-update,eq:weight-recurrence}.
The witness invariant identifies legal prefixes through $\val(W)$.
The representation invariant also covers prefixes with $\val(W)=0$ or
negative virtual counters. Every stored component represents a natural
number, so arithmetic observations terminate on every generated branch.
\begin{lemma}[All-candidate invariant]\label{lem:registers}
If $\Reg_m(\vec t;A,B,W)$ holds, each continuation of the corresponding
nonhalting state rule satisfies $\Reg_{m+1}$ for its updated natural
pairs. Hence every finite candidate prefix has such a representation.
\end{lemma}
\begin{proof}
The polynomial coordinates follow from \cref{lem:polynomial}.
An increment applies $\mathsf{Succ}$ to the positive coordinate;
a positive choice applies it to the negative coordinate and uses
$\mathsf{Pos},\mathsf{Neg}$ on the old weight and old counter.
A zero choice first applies $\Pol_m$ to the old counter's two numeral
terms and then uses the two multiplication components.
The unchanged coordinates retain their realizations. This proves the
step for each counter; induction from the seeds proves the last claim.
There is no legality or nonzero-weight hypothesis.
\end{proof}

\subsection{Operational correspondence, productivity, and disagreement}\label{sec:operational-correspondence}

Let $\mathcal K=b\,z\,\mathcal K\,\mathcal K$.
For an arbitrary triple $R=(A,B,W)\in(\N^2)^3$, write $W=(W^L,W^R)$:
$W^L$ is its first component and $W^R$ its second.
For $\diamond\in\{L,R\}$, define $T^\diamond(m,i,R)$
by its finite-address observations, or equivalently by the equations
\begin{equation}\label{eq:model-tree}
 T^\diamond(m,i,R)=
 \begin{cases}
 b\,(a^{W^\diamond}z)\,\mathcal K\,\mathcal K,
                                      &i\notin I_M,\\
 b\,z\,T^\diamond(m+1,i_1,R_1)\,
       T^\diamond(m+1,i_2,R_2),        &i\in I_M.
 \end{cases}
\end{equation}
The $i_e,R_e$ are the candidate successor and pair update for choice
$e\in\{1,2\}$; an increment duplicates its successor. The definition is
well founded when evaluated at a finite address, since a recursive
call consumes one address symbol. Every required finite node is
present: the root is $b$, observation children are finite chains,
and continuation children satisfy the same equations or are $\mathcal K$.

Every address is either $\epsilon$ or $iv$, with $i\in\N$ and
$v\in\N^*$. The three constructors act on observation maps by
\begin{equation}\label{eq:constructor-maps}
\begin{aligned}
 z(\epsilon)&=z,\\
 (aT)(\epsilon)&=a,
 &(aT)(0v)&=T(v),\\
 (bT_0T_1T_2)(\epsilon)&=b,
 &(bT_0T_1T_2)(iv)&=T_i(v)\qquad(i<3).
\end{aligned}
\end{equation}
All unlisted cases have value $\dagger$.
They realize these functions under ordinary terminal-head exposure.
The intermediate semantic domain contains arbitrary observation maps,
but the chains and the trees in \cref{eq:model-tree} satisfy the
rank conditions. Induction on finite depth establishes their totality.
The operational correspondence transfers this totality to the compiled
terms under ordinary HORS head rewriting.

\begin{lemma}[Operational correspondence]\label{lem:state}
If $\Reg_m(\vec t;R)$, then, for every address $w$ and terminal $a$,
\[
 \Obs_{G^\diamond_M}(\mathsf{F}_{\widehat i}\,\vec t,w,a)
       \iff T^\diamond(m,i,R)(w)=a .
\]
\end{lemma}
\begin{proof}
First, induction on the address proves that the rule for $\mathsf{K}$ realizes
$\mathcal K$. For the state assertion, induct on $w$, universally over
$m,i,\vec t,R$. Unfold the state once. At the empty address both sides
have root $b$. At a child index at least three neither side has an
observation. Child zero is $z$ in a nonhalting state; in a halting
state, \cref{lem:numerals} identifies it with the selected finite
chain. At child one or two of a nonhalting state, apply
\cref{lem:registers} and the induction hypothesis to the shorter
address. At a halt these children use the assertion for $\mathsf{K}$.
Normalization of $i$ preserves instruction lookup, including halts.
\end{proof}

Exposing $b$ gives the control node. Productivity also requires
its observation child to finish evaluating the selected arithmetic
term; \cref{lem:registers,lem:numerals} establish its termination.
\begin{lemma}[Productivity]\label{lem:productive}
Both compiled schemes are productive.
\end{lemma}
\begin{proof}
Prove the stronger assertion for every represented state by induction
on observation depth. Depth zero is immediate. A state exposes $b$
after one head unfolding. At child zero, the nonhalting case is $z$,
and the halting case realizes a finite chain by \cref{lem:numerals};
a chain produces every required node to the remaining depth.
The continuation children satisfy the induction hypothesis with the
updated invariant, or use $\mathsf{K}$, whose depth productivity follows from
its rule by the same induction. Apply this assertion to the initial
representation and the start rule. This proves productivity even on
illegal candidate paths and when the signed weight is zero.
\end{proof}

\begin{theorem}[Intermediate compilation]\label{thm:compiler}
The construction $M\mapsto(G_M^L,G_M^R)$ is effective and produces
well-formed safe productive schemes of order at most four over $\Sigma_0$.
Moreover, writing $\mathsf{S}$ for the start symbol in each scheme,
\[
\begin{aligned}
 M\halts\iff\exists w\in\N^*\;\exists a,b\in\Sigma_0\;
 \bigl[a\ne b &\land\Obs_{G_M^L}(\mathsf{S},w,a)\\
              &\land\Obs_{G_M^R}(\mathsf{S},w,b)\bigr].
\end{aligned}
\]
\end{theorem}
\begin{proof}
The static and productivity claims are \cref{lem:typing,lem:safety,lem:productive}.
The compiler traverses the finite instruction list and inserts the
displayed fixed templates, so it is total computable.

A candidate prefix of length $m$ determines an address
$v\in\{1,2\}^m$, ending at its state in \cref{eq:model-tree}.
This follows by induction on the prefix, using the chosen continuation
child at each step; all preceding states are nonhalting.
A legal halt gives weight components $r\ne s$, by
\cref{lem:live-halt} and $\val(W)=r-s$.
The two trees therefore differ at
\begin{equation}\label{eq:discrepancy-address}
                  v\,0\,0^{\min(r,s)}.
\end{equation}
One chain ends with $z$ there, while the other still has label $a$.
By \cref{lem:state}, the compiled HORS have this observation.

Conversely, consider any common address with different terminal labels.
As long as the address follows children one or two, both roots are $b$
and both computations have the same candidate control and pair data.
If such a path enters padding after a halt, the two subtrees coincide.
An out-of-range child cannot give a common observation.
Consequently a discrepancy must select child zero of a genuine
candidate state. At a nonhalting state both observation subtrees are $z$.
At a halt they are $a^rz$ and $a^sz$, whose only possible label mismatch
at a common address occurs at suffix $0^{\min(r,s)}$, with $r\ne s$.
Thus the candidate halt has nonzero weight and is a legal source halt
by \cref{lem:live-halt}.

Finally, two unequal productive ranked trees have different labels at
some address present in both trees: choose a shortest address of
disagreement. Its proper ancestors agree, and their labels determine
their child arities, so both nodes at that address exist. This also proves
the intermediate equivalence
$M\nonhalts\iff\BT(G_M^L)=\BT(G_M^R)$.
\end{proof}

\subsection{The two-terminal target}\label{sec:binary-target}

The final pass replaces the intermediate constructors by fixed contexts
over $\Sigma_2$. Define
\begin{equation}\label{eq:binary-encoding}
\begin{aligned}
 E(z)&=e,\\
 E(a\,t)&=d(e,E(t)),\\
 E(b(t_0,t_1,t_2))
   &=d\bigl(d(e,e),d(E(t_0),d(E(t_1),E(t_2)))\bigr).
\end{aligned}
\end{equation}
These guarded equations define $E$ on finite and infinite
$\Sigma_0$-trees.

\begin{lemma}[Binary output encoding]\label{lem:binary-encoding}
Every productive HORS $G$ over $\Sigma_0$ can be effectively translated
to a productive HORS $G^E$ over $\Sigma_2$ such that
\[
 \ord(G^E)\le\max\{1,\ord(G)\},\qquad
 \BT(G^E)=E(\BT(G)).
\]
If $G$ is safe, then $G^E$ is safe. For any productive HORS $G,H$
over $\Sigma_0$,
\[
              \BT(G^E)=\BT(H^E)\iff\BT(G)=\BT(H).
\]
\end{lemma}
\begin{proof}
The root label distinguishes $z$ from $a,b$; in the latter case,
the left child's label distinguishes $a$ from $b$.
The encoded child of $a$ occurs at address $1$, and the encoded
children of $b$ occur at $10,110,111$.
Decoding along finite addresses therefore gives
\[
                  E(T)=E(U)\quad\Longleftrightarrow\quad T=U.
\]

For a scheme $G$, rename its nonterminals to reserve three fresh
names, then replace its terminals $z,a,b$ by
$\mathsf Z:o$, $\mathsf A:o\to o$, and
$\mathsf B:o^3\to o$, with rules
\[
\begin{aligned}
 \mathsf Z&\to e,\\
 \mathsf A\,x&\to d\,e\,x,\\
 \mathsf B\,x_0\,x_1\,x_2
    &\to d\,(d\,e\,e)\,(d\,x_0\,(d\,x_1\,x_2)).
\end{aligned}
\]
Call the result $G^E$. The replacements have the original terminal
types, so all old subterm types and parameter occurrences are preserved.
The new helpers have order at most one, giving the claimed order bound.
Their types are homogeneous, and every operand containing a parameter
has ground type. Thus safety is preserved.

The translation commutes with substitution. Each exposed source
constructor is implemented by one helper unfolding and a finite
terminal context, giving $\BT(G^E)=E(\BT(G))$.
Every finite target observation visits finitely many source nodes and
context nodes, so productivity is preserved.
Injectivity of $E$ now gives both directions of equivalence preservation.
\end{proof}

Set $\widehat G_M^\diamond=(G_M^\diamond)^E$ for
$\diamond\in\{L,R\}$. By \cref{thm:compiler,lem:binary-encoding},
these are safe, productive HORS of order at most four over $\Sigma_2$,
and
\[
 M\nonhalts\iff\BT(G_M^L)=\BT(G_M^R)
       \iff\BT(\widehat G_M^L)=\BT(\widehat G_M^R).
\]
Both passes are total syntax transformations. Extend the compiler to
natural-number codes by mapping malformed source codes, which are
excluded from the halting set, to two copies of $\mathsf{S}\to e$
over $\Sigma_2$. This establishes the total computable reduction in
\cref{eq:main-reduction}.
The reduction from the halting complement yields $\PiOne$-hardness for
$\EqSafeProd$, and for the classes obtained by omitting safety,
productivity, or the order restriction. The matching upper bound is proved next.

%% file: figures/control-compiler.tex
\begin{figure*}[p]
\centering
\begin{minipage}{\textwidth}
\normalsize
\setlength{\abovedisplayskip}{5pt}
\setlength{\belowdisplayskip}{5pt}
\setlength{\jot}{2pt}
\hrule
\smallskip
\textbf{Rule context and shared abbreviations}
\[
\begin{gathered}
 L=|M|,\qquad \vec t=(p,q,u,v,r,s,h,k),\qquad
 p,q:\PP,\quad u,v,r,s,h,k:\NN,\\
 p'=\mathsf{Next}_+\,p\,q,\qquad
 q'=\mathsf{Next}_-\,p\,q,\qquad
 \widehat j=\begin{cases}j,&j\in I_M,\\0,&j\notin I_M,\end{cases}
 \qquad w_L=h,\quad w_R=k.
\end{gathered}
\]
\smallskip\hrule\smallskip
\textbf{Control, halt, padding, and start rules}\quad $\diamond\in\{L,R\}$.
\[
\begin{alignedat}{2}
 \mathsf{F}_i\,p\,q\,u\,v\,r\,s\,h\,k
  &\to b\,z\,
       \bigl(\mathsf{F}_{\widehat{i+1}}\,p'\,q'\,(\mathsf{Succ}\,u)\,v\,r\,s\,h\,k\bigr)
  &&\quad M(i)=\mathsf{inc}_A\\[-1pt]
  &\phantom{{}\to b\,z\,}
       \bigl(\mathsf{F}_{\widehat{i+1}}\,p'\,q'\,(\mathsf{Succ}\,u)\,v\,r\,s\,h\,k\bigr)
  &&\\[5pt]
 \mathsf{F}_i\,p\,q\,u\,v\,r\,s\,h\,k
  &\to b\,z\,
       \bigl(\mathsf{F}_{\widehat{i+1}}\,p'\,q'\,u\,v\,(\mathsf{Succ}\,r)\,s\,h\,k\bigr)
  &&\quad M(i)=\mathsf{inc}_B\\[-1pt]
  &\phantom{{}\to b\,z\,}
       \bigl(\mathsf{F}_{\widehat{i+1}}\,p'\,q'\,u\,v\,(\mathsf{Succ}\,r)\,s\,h\,k\bigr)
  &&\\[5pt]
 \mathsf{F}_i\,p\,q\,u\,v\,r\,s\,h\,k
  &\to b\,z\,
       \bigl(\mathsf{F}_{\widehat{i+1}}\,p'\,q'\,u\,v\,r\,s
  &&\quad M(i)=\mathsf{dec}_A(j)\\[-1pt]
  &\qquad(\mathsf{Pos}\,h\,k\,(p\,u\,v)\,(q\,u\,v))
  &&\\[-1pt]
  &\qquad(\mathsf{Neg}\,h\,k\,(p\,u\,v)\,(q\,u\,v))\bigr)
  &&\\[-1pt]
  &\phantom{{}\to b\,z\,}
       \bigl(\mathsf{F}_{\widehat j}\,p'\,q'\,u\,(\mathsf{Succ}\,v)\,r\,s
  &&\\[-1pt]
  &\qquad(\mathsf{Pos}\,h\,k\,u\,v)\,
           (\mathsf{Neg}\,h\,k\,u\,v)\bigr)
  &&\\[5pt]
 \mathsf{F}_i\,p\,q\,u\,v\,r\,s\,h\,k
  &\to b\,z\,
       \bigl(\mathsf{F}_{\widehat{i+1}}\,p'\,q'\,u\,v\,r\,s
  &&\quad M(i)=\mathsf{dec}_B(j)\\[-1pt]
  &\qquad(\mathsf{Pos}\,h\,k\,(p\,r\,s)\,(q\,r\,s))
  &&\\[-1pt]
  &\qquad(\mathsf{Neg}\,h\,k\,(p\,r\,s)\,(q\,r\,s))\bigr)
  &&\\[-1pt]
  &\phantom{{}\to b\,z\,}
       \bigl(\mathsf{F}_{\widehat j}\,p'\,q'\,u\,v\,r\,(\mathsf{Succ}\,s)
  &&\\[-1pt]
  &\qquad(\mathsf{Pos}\,h\,k\,r\,s)\,
           (\mathsf{Neg}\,h\,k\,r\,s)\bigr)
  &&\\[7pt]
 \mathsf{F}_0\,p\,q\,u\,v\,r\,s\,h\,k
  &\to b\,(w_\diamond\,a\,z)\,\mathsf{K}\,\mathsf{K}
  &&\\[4pt]
 \mathsf{K}&\to b\,z\,\mathsf{K}\,\mathsf{K}
  &&\\[4pt]
 \mathsf{S}&\to \mathsf{F}_{\widehat1}\,\mathsf{Seed}_+\,\mathsf{Seed}_-
  &&\\[-1pt]
  &\qquad\mathsf{Zero}\,\mathsf{Zero}\,\mathsf{Zero}\,\mathsf{Zero}\,
           \mathsf{One}\,\mathsf{Zero}.
  &&
\end{alignedat}
\]
\smallskip\hrule\smallskip
\textbf{Scheme assembly}\quad $\diamond\in\{L,R\}$.
\[
\begin{aligned}
 \llbracket M\rrbracket^\diamond
   &=(\Sigma_0,\mathcal N_M,\mathcal R_M^\diamond,\mathsf{S}),\\
 \mathcal N_M
   &=\mathcal N_{\mathrm{arith}}\cup\{\mathsf{S}:o,\mathsf{K}:o\}
     \cup\{\mathsf{F}_i:\PP\to\PP\to\NN^6\to o\mid 0\le i\le L\},\\
 \mathcal R_M^\diamond
   &=\mathcal R_{\mathrm{arith}}
     \cup\mathcal R_{\mathrm{ctrl}}(M)
     \cup\mathcal R_{\mathrm{aux}}^\diamond .
\end{aligned}
\]
\smallskip\hrule
\end{minipage}
\caption{Compilation of source control flow into the intermediate HORS
$G_M^\diamond=\llbracket M\rrbracket^\diamond$.
$\mathcal N_{\mathrm{arith}},\mathcal R_{\mathrm{arith}}$ are the
arithmetic nonterminals and rules of \cref{sec:rules}.
The first four templates select one rule per $i\in I_M$ by the
right-hand conditions, forming $\mathcal R_{\mathrm{ctrl}}(M)$;
the last three rules form $\mathcal R_{\mathrm{aux}}^\diamond$.
Primed terms, $\vec t$, and $w_\diamond$ are textual abbreviations
expanded before the rules are emitted.}
\label{fig:control-compiler}
\Description{A compiler figure collecting every control-flow rule.
The four source-instruction conditions are aligned to the right of
the corresponding F rules, followed by the halt F0, padding K and
start S rules. Finite scheme assembly appears below.}
\end{figure*}

%% file: sections/05-upper-bound.tex
\section{\texorpdfstring{$\PiOne$}{Π⁰₁}-completeness}\label{sec:upper}

Jochems provides an effective $\bot$-elimination transform that replaces
every $\bot$ node by an infinite chain of a fresh terminal, preserving
and reflecting tree equality~\cite[Lemma~2.7 and Proposition~2.8]{jochems2021}.
Its outputs are productive, so enumerating their finite prefixes
semidecides inequality and yields the general co-c.e.\ upper bound.

Here we give a direct proof using finite monotone semantics for head
convergence~\cite[Section~4.1 and Theorem~4.1]{salvati2015}.
We construct the finite model and prove
that it decides whether a closed ground term exposes a terminal.
This allows us to compute the original tree's prefix at any prescribed
finite depth, distinguishing terminal labels, $\bot$ nodes, and absent
addresses, even for nonproductive schemes. Equality and productivity
are then universal conditions on these computable prefixes. Together
with the decidable syntactic restrictions and the lower bound in
\cref{sec:binary-target}, this establishes the $\PiOne$-completeness
claims of \cref{thm:main}.

\subsection{A finite semantics for head convergence}

Our immediate task is to decide whether a closed ground term $t:o$
exposes a terminal after finitely many head steps. This concerns only
the current root: a scheme with the rule $\mathsf{S}\to a\,\mathsf{S}$
exposes $a$, although it generates an infinite tree, whereas
$\mathsf{S}\to\mathsf{S}$ never exposes a terminal.

We interpret recursive definitions by least fixed points, following
the standard domain-theoretic construction~\cite[Theorem~2.1.19]{abramsky1994}.
We assign each type $\sigma$ a finite semantic domain $D_\sigma$.
At ground type we retain only convergence information, using
\[
 D_o=\{0<1\}.
\]
The value $1$ records head convergence; $0$ is the least approximation,
providing no convergence guarantee. The order $0<1$ expresses increasing
information. Labels and descendants are forgotten: $z$ and $a\,z$
both have value $1$, as does $a\,u$ even when $u$ diverges.

For a function type, the abstract value describes how argument
information determines result information. Define $D_{\sigma\to\tau}$
to consist of all monotone maps $D_\sigma\to D_\tau$, ordered pointwise.
Monotonicity ensures that increasing argument information cannot
decrease result information. For example,
\[
 D_{o\to o}=\{d\mapsto0,\ d\mapsto d,\ d\mapsto1\}.
\]
These describe, respectively, an always-divergent result, returning the
argument, and an always-convergent head such as a fixed terminal.
Each domain is finite by induction on types, since there are only
finitely many maps between finite sets. Its least element is
$\bot_o=0$ at ground type and the constant-$\bot_\tau$ map at
type $\sigma\to\tau$.

Let $F:\tau_F$ range over the input scheme $G$'s finite list of nonterminals, each
with its given rule $F\,x_1\cdots x_k\to t_F$.
A global environment $\eta$ assigns a current approximation
$\eta(F)\in D_{\tau_F}$ to every such symbol. All these environments
form the finite poset $E_G=\prod_{F:\tau_F}D_{\tau_F}$, ordered pointwise.

The parameters of a rule receive values locally. For a rule with
$x_i:\sigma_i$ and each tuple of abstract arguments
$\vec v\in D_{\sigma_1}\times\cdots\times D_{\sigma_k}$, form
\[
 \rho_{\vec v}=[x_1\mapsto v_1,\ldots,x_k\mapsto v_k].
\]
A rule without parameters uses the empty local environment.
For given global and local environments, define the term denotation
$\ev{t}_{\eta,\rho}$ recursively by
\begin{equation}\label{eq:finite-denotation}
\begin{aligned}
 \ev{x}_{\eta,\rho}&=\rho(x),\\
 \ev{F}_{\eta,\rho}&=\eta(F),\\
 \ev{a}_{\eta,\rho}&=\lambda d_1\cdots d_r.\,1
                                      &&(a:o^r\to o),\\
 \ev{t\,u}_{\eta,\rho}&=
       \ev{t}_{\eta,\rho}\bigl(\ev{u}_{\eta,\rho}\bigr).
\end{aligned}
\end{equation}
For a nullary terminal the third clause gives $1$. For closed terms
we omit $\rho$, since there are no free parameters.
Typing induction shows that $\ev{t}_{\eta,\rho}\in D_\sigma$ whenever
$t:\sigma$. The interpretation is also monotone in both environments.
For application, increasing the function and argument values from
$f,d$ to $g,e$ gives
$(f\le g\land d\le e)\Longrightarrow f(d)\le f(e)\le g(e)$:
the first inequality uses monotonicity of $f$, and the second the
pointwise order $f\le g$.

The rules determine a single update of the global environment:
\[
 \Phi_G(\eta)(F)(\vec v)=\ev{t_F}_{\eta,\rho_{\vec v}}.
\]
Typing and monotonicity make $\Phi_G$ a monotone endomap of $E_G$.
We next compute its least fixed point by finite iteration.

\subsection{Effective computation of the least fixed point}\label{sec:finite-tables}

To compute with this model, represent each domain by a list of distinct
elements with decidable order and a least element.
At ground type use the list $[0,1]$ with $0<1$.
Given representations of $D_\sigma,D_\tau$, enumerate the
$|D_\tau|^{|D_\sigma|}$ function tables and retain exactly those satisfying
\[
 \forall x,y\in D_\sigma\quad
       x\le y\Longrightarrow f(x)\le f(y).
\]
This finite test enumerates $D_{\sigma\to\tau}$. Equality and order
are pointwise, application is lookup, and the constant-$\bot_\tau$
table is least.

To compute $\Phi_G(\eta)$, enumerate each rule's abstract argument tuples
and evaluate its body by \cref{eq:finite-denotation}; higher-type
arguments are finite function tables. Grouping the results by successive
arguments constructs the curried table for each nonterminal. Thus the
update, equality, and cardinality of the finite product $E_G$ are computable.

Use the standard ascending fixed-point iteration
of abstract interpretation~\cite[Section~8]{cousot1977}:
initialize $\eta_0(F)=\bot_{\tau_F}$ for every declared nonterminal and
iterate $\eta_{n+1}=\Phi_G(\eta_n)$. Each round updates all entries
simultaneously, using the same old environment $\eta_n$ throughout.
For example, take $\mathsf{S}:o$ and $\mathsf{F}:o\to o$ with rules
$\mathsf{S}\to\mathsf{F}\,z$ and $\mathsf{F}\,x\to x$.
For each $d\in D_o$, the local environment is $[x\mapsto d]$, so
\[
\begin{aligned}
 \Phi_G(\eta)(\mathsf{F})(d)
   &=\ev{x}_{\eta,[x\mapsto d]}=d,\\
 \Phi_G(\eta)(\mathsf{S})&=\eta(\mathsf{F})(1).
\end{aligned}
\]
The pair $(\eta_n(\mathsf{S}),\eta_n(\mathsf{F}))$ therefore evolves as
\[
 (0,\,d\mapsto0)\;\longmapsto\;(0,\,d\mapsto d)
 \;\longmapsto\;(1,\,d\mapsto d)\;\longmapsto\;(1,\,d\mapsto d).
\]
The start value stays $0$ in the first update, while the function table
changes. This shows why an unchanged start value alone cannot justify stopping.

Stop when $\eta_{n+1}=\eta_n$, comparing all entries of all function
tables. Subsequent environments then remain equal. Since $\eta_0$ is
least and $\Phi_G$ is monotone, the sequence is increasing and stabilizes
by $B=|E_G|$: the chain cannot contain $B+1$ distinct environments.
Hence one can stop at the first equality or compute
exactly $B$ rounds; write $\eta_*=\eta_B$.
For any fixed point $\theta$, induction gives $\eta_n\le\theta$ for
every $n$, so $\eta_*$ is the least fixed point.
This computes $\ev{t}_{\eta_*}$ for every closed typed term $t$.
The next subsection proves that its ground value distinguishes head
convergence from divergence.

\subsection{Adequacy for head convergence}\label{sec:head-adequacy}

The preceding construction computes $\ev{t}_{\eta_*}$.
We now prove that, for a closed ground term, this value is $1$ exactly
when actual head unfolding exposes a terminal after finitely many steps.

\begin{lemma}[Head adequacy]\label{lem:head}
For a well-formed scheme and a closed ground term $t$,
\[
 \ev{t}_{\eta_*}=1
 \quad\Longleftrightarrow\quad
 \exists a,\vec u\;t\Downarrow_G a(\vec u).
\]
\end{lemma}
\begin{proof}
We first record that substitution preserves typing. Induction on the
term proves this: variables take the types of their substituted
arguments, constants are unchanged, and application uses the induction
hypotheses. Thus head unfolding preserves typing and exposes well-typed
children.

\begin{enumerate}[label=\textbf{(\arabic*)},leftmargin=*,itemsep=0.8em]
\item \textbf{Soundness ($\Rightarrow$).}
We must show that a semantic value of $1$ has a finite operational
witness. To carry this guarantee through higher-order application,
we relate terms to abstract values: at ground type, value $1$ requires
head convergence; at function types, related arguments must produce
related results. For closed $t:\sigma$ and $v\in D_\sigma$, define
$t\mathrel{\mathcal R_\sigma}v$ by
\[
\begin{aligned}
 t\mathrel{\mathcal R_o}v
   &\iff v=1\Longrightarrow\exists a,\vec u\;t\Downarrow_G a(\vec u),\\
 t\mathrel{\mathcal R_{\sigma\to\tau}}f
   &\iff\forall u,v\;
       [u\mathrel{\mathcal R_\sigma}v
                     \Longrightarrow t\,u\mathrel{\mathcal R_\tau}f(v)].
\end{aligned}
\]
Every closed typed term is related to the least element at its type.
This follows by induction on the type: the base implication for zero
imposes no convergence claim, and a bottom function always returns a
bottom value.
A terminal is related to its constant-one interpretation, since its
saturated application exposes that terminal without evaluating its
children.

The relation is preserved by interpreting terms. More precisely, for a
term $t:\sigma$, if every nonterminal $F$ is related to $\eta(F)$ and each closed argument
$u_i$ is related to $\rho(x_i)$, then
\[
 t[\vec u/\vec x]\mathrel{\mathcal R_\sigma}\ev{t}_{\eta,\rho}.
\]
Induction on $t$ proves it: lookup handles variables and nonterminals,
the preceding observation handles terminals, and the function clause
handles application.

We now induct on the iteration number $n$ to show
$F\mathrel{\mathcal R_{\tau_F}}\eta_n(F)$ for every nonterminal.
The base case follows from $\eta_0(F)=\bot_{\tau_F}$.
For the step, take arbitrary closed arguments $u_i$ related to values
$v_i$. The substitution property relates the instantiated body
$t_F[\vec u/\vec x]$ to $\Phi_G(\eta_n)(F)(\vec v)$.
Prepending the actual head step
$F\,\vec u\to_h t_F[\vec u/\vec x]$ preserves the ground relation:
any head exposure of the body is also a head exposure of the call.
Quantifying over the arguments gives the required function relation
for $F$ at $\eta_{n+1}(F)$.
At $n=B$, all nonterminals are therefore related to their values in
$\eta_*$. Applying the substitution property to the closed ground
term $t$ gives $t\mathrel{\mathcal R_o}\ev{t}_{\eta_*}$.
If this value is $1$, the ground clause yields the required finite
head exposure.

\item \textbf{Completeness ($\Leftarrow$).}
Suppose a finite head-unfolding sequence takes $t$ to a terminal
application. We show that each head step preserves the interpretation
in $\eta_*$. Semantic substitution, proved by
structural induction, and the fixed-point equation
$\eta_*=\Phi_G(\eta_*)$ give, for each rule $F\,\vec x\to t_F$
and matching closed arguments $\vec u$,
\[
\begin{aligned}
 \ev{F\,\vec u}_{\eta_*}
   &=\eta_*(F)(\ev{\vec u}_{\eta_*})\\
   &=\ev{t_F}_{\eta_*,[\vec x\mapsto\ev{\vec u}_{\eta_*}]}\\
   &=\ev{t_F[\vec u/\vec x]}_{\eta_*}.
\end{aligned}
\]
Here the denotation of a tuple is taken componentwise.
The first equality is application, the second uses the fixed point
and the definition of $\Phi_G$, and the third is semantic substitution.
The terminal application at the end of the finite head-unfolding
sequence has value $1$, independently of its children.
Induction along that sequence therefore gives $\ev{t}_{\eta_*}=1$,
proving the reverse implication.\qedhere
\end{enumerate}
\end{proof}

\subsection{Exact finite prefixes}

Define $\mathsf{head}_G(t)$ as follows. If $\ev{t}_{\eta_*}=0$, return
$\bot$. Otherwise run deterministic head unfolding until a terminal
$a(\vec u)$ is reached, and return that terminal and its arguments.
The search terminates by \cref{lem:head} and requires no a priori
bound on the number of head steps.
Typing is preserved, so each returned child can be processed in turn.

Head search can try successive finite unfolding budgets, since
success at any budget is decidable.
It is invoked only after head adequacy has proved that some budget
succeeds. This gives a total computable head reader. On a negative
finite-model answer, no search is performed.

The exact address reader $\mathsf{read}_G(t,w)$ returns a terminal label
$a\in\Sigma$ at an existing node, $\bot$ at a present divergent node,
and $\dagger$ when the address is absent. Define it by recursion on $w$:
\[
\begin{aligned}
 \mathsf{read}_G(t,\epsilon)&=
 \begin{cases}
  \bot,&\mathsf{head}_G(t)=\bot,\\
  a,&\mathsf{head}_G(t)=a(\vec u),
 \end{cases}\\[0.6em]
 \mathsf{read}_G(t,iw')&=
 \begin{cases}
  \mathsf{read}_G(u_i,w'),
     &\mathsf{head}_G(t)=a(\vec u)\land i<|\vec u|,\\
  \dagger,&\text{otherwise}.
 \end{cases}
\end{aligned}
\]
The recursion terminates by address length, with each head call total.

\begin{lemma}[Exact prefixes]\label{lem:prefix}
For every well-formed $G$, closed ground $t$, and finite address $w$,
the reader returns exactly the label, bottom node,
or absence at address $w$ in $\BT_G(t)$. Consequently exact prefixes
of any specified finite depth are total computable.
\end{lemma}
\begin{proof}
Induct on $w$. At $\epsilon$, \cref{lem:head} determines convergence
and head unfolding supplies the terminal label. At $iw'$, bottom has
no children, invalid terminal-child indices are absent, and valid
indices invoke the induction hypothesis.
Let $b$ be the maximum terminal arity, or zero for an empty alphabet.
Reading all $w\in\{0,\ldots,b-1\}^*$ with $|w|<d$ in a fixed order
computes $\Pref(G,d)$. Every actual node has such an address, so
equality of all prefix tables is equivalent to tree equality.
\end{proof}

Finally, define the Boolean depth test by
\[
\begin{aligned}
 Q(G,t,0)&=\mathsf{true},\\
 Q(G,t,d+1)&=
 \begin{cases}
  \mathsf{false},&\mathsf{head}_G(t)=\bot,\\
  \displaystyle\bigwedge_{i<r}Q(G,u_i,d),
     &\mathsf{head}_G(t)=a(u_0,\ldots,u_{r-1}).
 \end{cases}
\end{aligned}
\]
Induction on $d$ identifies this total test with depth-$d$ productivity.
A nullary terminal passes because the conjunction is empty.
For start terms abbreviate it to $Q(G,d)$. Hence
\[
\begin{split}
 \langle G,H\rangle\in\EqProd\iff
 \forall d\;[&
 \Valid_{4,\Sigma_2}(G,H)\ \land\
 \Pref(G,d)=\Pref(H,d)\\
 &{}\land Q(G,d)\land Q(H,d)] .
\end{split}
\]
On invalid input the matrix is defined to be false before any semantic
operation. This is a computable universal Boolean definition.
Conjoining the decidable checks $\Safe(G)\land\Safe(H)$ gives the
same upper bound for $\EqSafeProd$. Omitting the $Q$ tests or the
syntactic restrictions gives the bounds for $\EqFour$ and $\EqH$.
Together with the reduction in \cref{sec:binary-target}, this proves
\cref{thm:main}.

%% file: sections/07-consequences.tex
\Needspace{12\baselineskip}
\section{Consequences}\label{sec:consequences}

For pushdown automata we use the ranked-tree convention
of~\cite[Section~3.3]{hague2017}: a deterministic higher-order automaton
without collapse, together with its output labelling, generates a tree
by unfolding transitions and contracting silent steps. Infinite silent
runs give $\bot$ leaves.

\begin{corollary}[Pushdown-tree equivalence]\label{cor:pda-tree}
Equality of the labelled trees generated by deterministic order-four
higher-order pushdown automata without collapse under this convention
is undecidable.
\end{corollary}
\begin{proof}
The effective, order-preserving translation
of~\cite[Theorem~2, p.~14]{hague2017} converts each safe final compiled scheme into
an order-four automaton with the same ordered labelled tree.
The convention's one-based child indices correspond to our zero-based
indices by uniform reindexing. By \cref{eq:main-reduction}, equality of the two
translated trees is equivalent to source nonhalting. A decision
procedure for these automata would therefore decide source halting.
\end{proof}

\begin{corollary}[$\lambda Y$ B\"ohm-tree equivalence]\label{cor:lambda-y}
Equality of the $\eta$-long B\"ohm trees of simply typed $\lambda Y$-terms
up to $\alpha$-equivalence is $\PiOne$-complete under total computable
many-one reductions.
Inputs are pairs of terms of the same type in a common typed context.
\end{corollary}
\begin{proof}
The translations of~\cite[Theorem~11 and Corollary~19, pp.~96, 101]{clairambault2013}
give computable many-one reductions in both directions between HORS
equivalence and this problem. The B\"ohm-tree encoding preserves binding
information through de Bruijn levels, using a common enumeration of free
variables. The HORS-to-$\lambda Y$ direction transfers $\PiOne$-hardness
from \cref{thm:main}. The reverse direction transfers its $\PiOne$
upper bound.
\end{proof}

\paragraph{Limits of effective comparison.}
The main theorem also rules out complete effective methods for
certifying equality or bounding the depth of disagreement.

\begin{corollary}\label{cor:certificates}
No sound and complete effective system of finite certificates recognizes
tree equality throughout the safe productive order-four class,
even over $\Sigma_2$.
\end{corollary}
\begin{proof}
Enumerating and checking the certificates would semidecide equality.
Exact-prefix search semidecides inequality. Dovetailing these searches
would decide equality on the final compiler's outputs, contradicting
\cref{eq:main-reduction} and source undecidability.
\end{proof}

\begin{corollary}\label{cor:depth}
There is no total computable function $B$ of input-description size
such that every unequal safe productive order-four pair over $\Sigma_2$
has a discrepancy at an address of length at most $B(n)$.
\end{corollary}
\begin{proof}
Compare the depth-$B(n)+1$ prefixes of the final compiled pair of size $n$.
A difference implies source halting; agreement forces tree equality
and hence nonhalting. This would decide the source halting problem.
\end{proof}

\paragraph{The order-three boundary.}
The order-four bound arises from the stored polynomial interface. While arithmetic helpers
have order at most three, $\NN$ has order two; hence $\PP = \NN \to \NN \to \NN$ has order
three, which pushes $\mathsf{Next}_\pm$ and the control nonterminals $\mathsf{F}_i$ to order
four (\cref{lem:typing}). Inlining update helpers cannot eliminate this bottleneck.
Reducing the target order to three would require either storing polynomial representations
of order at most two capable of subsequent updates and evaluations, or bypassing stored
polynomials entirely. We therefore leave HORS equivalence at orders two and three open
for future work.

%% file: sections/06-formalization.tex
\paragraph{Formalization.}
A supplementary Rocq development covers both compiler stages, including
typing, safety, the order bound, productivity, and correspondence with
standard HORS head rewriting, as well as the finite semantics and
observation algorithms used for the upper bound. It relies on the Rocq
standard library and the Library of Undecidability Proofs~\cite{undecidability},
with completeness formulated in the synthetic computability framework
of~\cite{forster2019}; conventional effectivity is established by the
paper's explicit algorithms. The consequences in \cref{sec:consequences},
including the pushdown-automata translation, remain paper-level deductions.
Codex was used as an interactive aid during formalization development
(e.g., for boilerplate and proof scripts) and for revising the paper's prose.
All formal definitions, theorems, and proofs were audited and verified by
the authors. The underlying theory was developed independently by the authors.
The artifact is available at 
\href{https://doi.org/10.5281/zenodo.23734555}{\texttt{doi:10.5281/zenodo.23734555}}; 
its \texttt{README} provides the precise theorem statements, assumptions, and 
reproduction instructions.